\documentclass{aims}

\usepackage{amsmath}
\usepackage{mathrsfs}
\usepackage{bbm}

  \usepackage{paralist}
  \usepackage{graphics} 
  \usepackage{epsfig} 
\usepackage{graphicx}  \usepackage{epstopdf}
 \usepackage[colorlinks=true]{hyperref}
\hypersetup{urlcolor=blue, citecolor=red}

\def\currentvolume{X}
 \def\currentissue{X}
  \def\currentyear{200X}
   \def\currentmonth{XX}
    \def\ppages{X--XX}
     \def\DOI{10.3934/xx.xx.xx.xx}

\newtheorem{theorem}{Theorem}[section]
\newtheorem{corollary}{Corollary}

\newtheorem{lemma}[theorem]{Lemma}
\newtheorem{proposition}{Proposition}

\theoremstyle{definition}
\newtheorem{definition}[theorem]{Definition}

\newcommand{\diff}{\mathrm{d}}
\newcommand{\e}{\operatorname{e}}
\newcommand{\Tr}{\operatorname{Tr}}
\newcommand{\MM}{\operatorname{\mathcal{M}}}

\newcommand{\NN}{\operatorname{\mathcal{N}}}

\title[Quantum Radon Transform] 
      {Quantum Tomography and the Quantum Radon Transform}

\author[Alberto Ibort and Alberto L\'opez-yela]{}


 \email{albertoi@math.uc3m.es}
 \email{alyela@math.uc3m.es}

\thanks{$^*$ Corresponding author: Alberto Ibort}

\begin{document}
\maketitle

\centerline{\scshape Alberto Ibort$^*$}
\medskip
{\footnotesize
 \centerline{Instituto de Ciencias
Matem\'{a}ticas (CSIC - UAM - UC3M - UCM) ICMAT and}
   \centerline{Depto. de Matem\'aticas, Univ. Carlos III de Madrid,}
   \centerline{ Avda. de la Universidad 30, 28911 Legan\'es, Madrid, Spain}
} 

\medskip

\centerline{\scshape Alberto L\'opez-Yela}
\medskip
{\footnotesize
 \centerline{ Dpto. de Teor\'ia de la se\~nal y comunicaciones, Univ. Carlos III de Madrid,}
   \centerline{Avda. de la Universidad 30, 28911 Legan\'es, Madrid, Spain}
}

\bigskip

 \centerline{(Communicated by the associate editor name)}

\begin{abstract}
A general framework in the setting of $C^*$-algebras for the tomographical description of states, that includes, among other tomographical schemes, the classical Radon transform, quantum state tomography and group quantum tomography, is presented.  

Given a $C^*$-algebra, the main ingredients for a tomographical description of its states are identified: A  generalized sampling theory and a positive transform.  A generalization of the notion of dual tomographic pair provides the background for a sampling theory on $C^*$-algebras and, an extension of Bochner's theorem for functions of positive type, the positive transform.

The abstract theory is realized by using dynamical systems, that is, groups represented on $C^*$-algebra.  Using a fiducial state and the corresponding GNS construction, explicit expressions for tomograms associated with states defined by density operators on the corresponding Hilbert spade are obtained.   In particular a general quantum version of the classical definition of the Radon transform is presented.    
The theory is completed by proving that if the representation of the group is square integrable, the representation itself defines a dual tomographic map and explicit reconstruction formulas are obtained by making a judiciously use of the theory of frames.
A few significant examples are discussed that illustrates the use and scope of the theory.
\end{abstract}

\tableofcontents

\section{Introduction}\label{sec:introduction}
In 1917 Johann Karl August Radon  \cite{Ra17} introduced the transformation that carries his name and that allows to recover a function (regular enough) from its averages over a family of lines.  The Radon transform provides the mathematical backbone of modern Computerized Axial Tomography (CAT for short, see Fig. \ref{Section_CAT}) \cite{Na86}.   

A similar idea was proposed \cite{Vo89}, \cite{Be87}, and implemented experimentally \cite{Sm93}, to reconstruct the Wigner function  associated with the state of a quantum system. The proposed method shares the essence of the tomographic method: performing many measurements through slices of the given function (the Wigner quasidistribution associated with the quantum state in this case) to obtain a marginal probability distribution (the tomogram),  from which the original object can be reconstructed by means of an generalized inverse Radon transform, hence the given name of \textit{quantum state tomography} to these techniques. 

More formally, let $f(q,p)$ be a Schwartz function on $\mathbb{R}^2$. The \textit{Radon Transform} of $f$ is defined as:
\begin{equation}\label{Radon_Transform}
\mathcal{W}_f(X,\mu,\nu) = \int_{l(X,\mu,\nu)} f\big(q(s),p(s)\big)\diff s = \int\limits_{\hspace{.5cm}\mathbb{R}^2}\hspace{-.19cm} f(q,p)\delta(X-\mu q-\nu p)\diff q\diff p \, ,
\end{equation}
where $l(X,\mu,\nu)$ denotes the  line $X - \mu q - \nu p = 0$, in the $(q,p)$ plane we integrate over, and  $\delta (X- \mu q - \nu p)$ the delta distribution along $l(X,\mu,\nu)$.  The original function $f$ can be recovered by means of the following expression \cite{Ma05}:
\begin{equation}\label{Inverse_Radon_Transform_2}
f(q,p)=\frac{1}{(2\pi)^2}\hspace{-.3cm}\int\limits_{\hspace{.5cm}\mathbb{R}^3}\hspace{-.2cm}\mathcal{W}_f(X,\mu,\nu)\e^{-i(X-\mu q-\nu p)}\diff X\diff\mu\diff\nu.
\end{equation}

\begin{figure}[htp]
\begin{center}
 \includegraphics[width=5in]{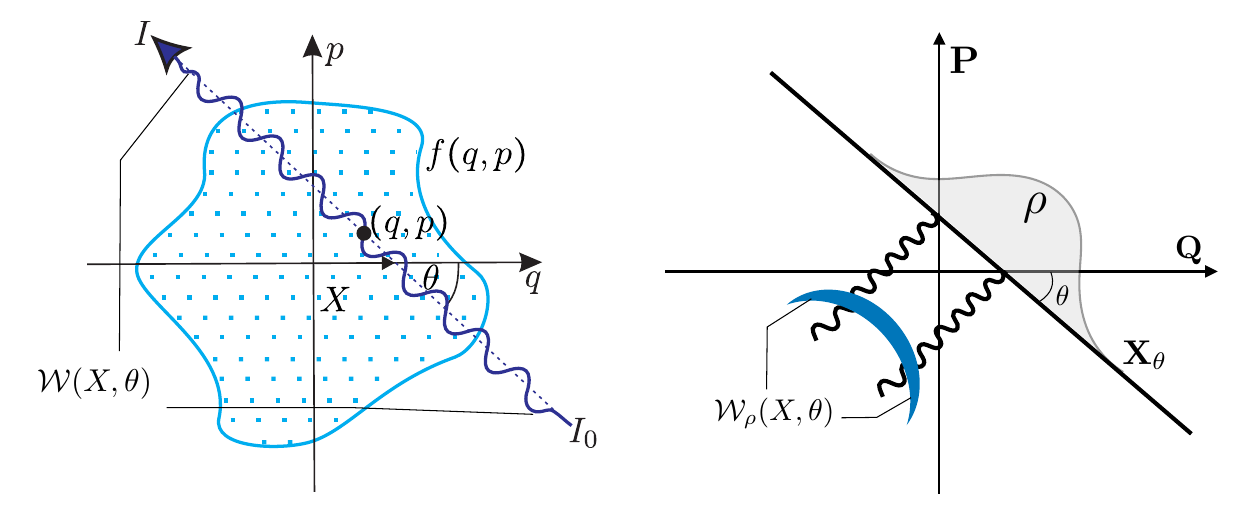}\\
  \caption{Left: Section of a body irradiated with gamma radiation and the measure of the tomogram $\mathcal{W}(X,\theta)$. Right:  Scheme of quantum state tomography. The quantum tomogram $\mathcal{W}_\rho (X, \theta)$ obtained by homodyne measuring the quadrature operator $\mathbf{X}_\theta = X -\textbf{Q}\cos\theta-\textbf{P}\sin\theta$.}\label{Section_CAT}
  \end{center}
\end{figure}

The state of a classical system (where, for simplicity, we will just consider a $2$-dimensional phase space with coordinates $(q,p)$) can be described by a classical probability distribution $w (q,p)$.   The expected value of an observable $A$, a bounded function of position and momentum $A = A(q,p)$, is given by:
\begin{equation}\label{eq:classical_average}
\langle A \rangle_w = \int A(q,p) w(q,p) \diff q \diff p \, . 
\end{equation}
We can use the Radon Transform to describe the state $w = w(q,p)$ (and to reconstruct it) from its associated tomograms: 
$$
\mathcal{W}_w(X,\mu, \nu ) = \int w(q,p) \delta (X - \mu q - \nu p) \diff q \diff p\, .
$$

If we consider a quantum system instead of a classical one, its states are described by density operators $\boldsymbol{\rho}$, that is, by positive normalized trace class operators on a certain Hilbert space $\mathcal{H}$.  Observables are self-adjoint operators $\mathbf{A}$ on such Hilbert space and the expected value of the observable $\textbf{A}$ in a state $\boldsymbol{\rho}$ is given by the trace of its projection onto the state: 
\begin{equation}\label{trace_measurement}
\langle \mathbf{A}\rangle_{\boldsymbol{\rho}} = \Tr\left(\boldsymbol{\rho}\mathbf{A} \right).
\end{equation}
One way to draw a relation between the classical average formula \eqref{eq:classical_average} and the quantum one \eqref{trace_measurement} will consist of associating a numerical function $A(q,p)$ to the quantum observable $\mathbf{A}$, and a suitable `distribution'  function $w_{\boldsymbol{\rho}} (q,p)$ to the density operator $\boldsymbol{\rho}$ such that the expectation value of the quantum observable would be given as:
$$
\langle \mathbf{A}\rangle_{\boldsymbol{\rho}}  = \int A(q,p) w_{\boldsymbol{\rho}} (q,p) \diff q \diff p \,.
$$ 

Surprisingly perhaps, this problem has a solution. It is given by the quasi-distribution  $w_{\boldsymbol{\rho}}$ introduced 
by E. Wigner  \cite{Wi32}, universally known as the Wigner distribution, with the aim of  studying the quantum corrections to classical statistical mechanics \cite{Er07} and given by:
\begin{equation}\label{Wigner_function}
w_{\boldsymbol{\rho}} (q,p) = \frac{1}{(2\pi\hbar)}\hspace{-.1cm}\int\limits_{\hspace{.3cm}\mathbb{R}}\hspace{-.1cm} \left\langle q-\frac{y}{2}\right|\boldsymbol{\rho}\left|q+\frac{y}{2}\right\rangle\e^{ipy/\hbar}\diff y \, ,
\end{equation}
The Radon Transform of the Wigner distribtution $w_{\boldsymbol{\rho}}$ is the tomogram $\mathcal{W}_{\boldsymbol{\rho}}$ mentioned above used in quantum state tomography:
\begin{equation}\label{Tomogram_quantum_to_classic_4}
\mathcal{W}_{\boldsymbol{\rho}}(X,\theta) = \hspace{-.26cm}\int\limits_{\hspace{.5cm}\mathbb{R}^2}\hspace{-.19cm}w_{\boldsymbol{\rho}}(q,p)\delta(X -q\cos\theta-p\sin\theta)\diff q\diff p,
\end{equation}
Remarkably enough (see Sect. \ref{section_tomograms_group}, Thm.\ref{tom_cart} below, for the explanation of this general fact), the quantum tomogram $\mathcal{W}_{\boldsymbol{\rho}}$ associated with the state $\boldsymbol{\rho}$ can be expressed as:
\begin{equation}\label{eq:quantum_radon_heisenberg}
\mathcal{W}_{\boldsymbol{\rho}} (X,\mu, \nu) = \Tr\!\big(\boldsymbol{\rho}\, \delta(X\mathbbm{1}-\textbf{Q}\cos\theta-\textbf{P}\sin\theta)\big) = \langle \delta(X\mathbbm{1}-\textbf{Q}\cos\theta-\textbf{P}\sin\theta)\rangle_{\boldsymbol{\rho}} \, ,
\end{equation}
where $\textbf{Q}$ and $\textbf{P}$ are the quantum position and momentum operators respectively (see Fig. \ref{Section_CAT}, right, for a schematic description of the experimental setting used to measure $\mathcal{W}_{\boldsymbol{\rho}}$), and $\tan \theta = \mu/ \nu$, making it more apparent the parallelism with the tomogram of a classical state introduced before:
\begin{equation}\label{Radon_quantization}
\mathcal{W}_w (\mu, \nu) =  \int w(q,p) \delta (X - \mu q - \nu p) \diff q \diff p =   \langle \delta(X-q\cos\theta-p\sin\theta)\rangle_w \, .
\end{equation} 
The Wigner function $w_{\boldsymbol{\rho}} (q,p)$ can also be easily recovered from the tomogram \eqref{Tomogram_quantum_to_classic_4}:
\begin{equation}\label{Inverse_Radon_Transform_Wigner}
w_{\boldsymbol{\rho}} (q,p) =\frac{1}{(2\pi)^2}\hspace{-0.1cm}\int \mathcal{W}_{\boldsymbol{\rho}}(X,\theta)\e^{-ik(X-q\cos\theta-p\sin\theta)}k\hspace{0.02cm}\diff\theta\diff k\diff X \, ,
\end{equation}
and the matrix elements  in coordinate representation of the density operator $\boldsymbol{\rho}$ are given by the Fourier Transform of the Wigner function:
\begin{equation}\label{matrix_element_rho}
\boldsymbol{\rho}(q,q')=\hspace{-.1cm}\int \e^{-ip(q'-q)/\hbar} w_{\boldsymbol{\rho}}\left(\frac{q'+q}{2},p\right)\diff p \, .
\end{equation}

A number of other different quantum tomographic techniques adapted to different systems, such as spin systems \cite{Ma97}, molecular vibrational states \cite{Op99}, modes of radiation fields \cite{Ar00}, ultracold gases \cite{Ca08}, etc., have been proposed.   For systems exhibiting a group symmetry a unified theory based on the theory of group representations, and consequently called \textit{group quantum tomography}, has been developed by D'Ariano \textit{et al} (see for instance \cite{Ar00b}, \cite{Ca00} and references therein).   More general situations where, for instance, no group theoretical background was available, have also been discussed and a mature mathematical and physical theoretical background has been laid (see for instance the monograph \cite{Da03}, the review \cite{Ib09}, and references therein).   

A systematic exploration of the mathematical background for the various quantum tomographic frameworks based on the so called quantizer-dequantizer formalism, has been conducted  and the underlying mathematical and physical problems have been described \cite{Ma05}, \cite{Ib10}, \cite{Fa10}, \cite{As15}, \cite{As15_2}.   From these efforts it emerges that an analysis based on the algebraic description of quantum systems, i.e., using the $C^*$-algebric picture of quantum mechanics, would be relevant \cite{Ib11}.  

The description of a physical system always involves the selection of its algebra of observables $\mathcal{O}$ and a family of states $\mathcal{S}$. The outputs of measuring a given observable $A\in\mathcal{O}$, when the system is in the state $\rho\in\mathcal{S}$, are described by a probability measure $\mu_{A,\rho}$ on the real line, such that $\mu_{A,\rho}(\Delta)$ is the probability that the output of $A$ belongs to the subset $\Delta\subset\mathbb{R}$. Thus, a measure theory (or better a theory of measurement) for the physical system under consideration, is a pairing between states $\rho$ and observables $A$, that assigns Borel probability measures $\mu_{A,\rho}$ to them. Then, the expected value of the observable $A$ in the state $\rho$ is given by:
\begin{equation}
\langle A\rangle_\rho= \int \lambda\, \diff\mu_{A,\rho}(\lambda).
\end{equation}
The $C^*$-algebraic description of quantum systems assumes that the algebra of observables of a given quantum system is a $C^*$-algebra $\mathcal{A}$ (actually observables are its self-adjoint elements) and the states of the system are the mathematical states (i.e., positive, normalized linear forms) of the $C^*$-algebra.

Thus the purpose of this work will be to develop a $C^*$-algebra based Quantum Tomography, aiming to provide, not only a unified picture of all previous tomographic analysis of quantum systems, but a natural extension of the Radon transform that could be used effectively in other areas of application.     

Thus, the main ingredients needed to construct a tomographic picture of a system in a $C^*$-algebraic framework, will be analysed:  A sampling theory on $C^*$-algebras and, a positive, Bochner-like, transform. The objective of the first will be to provide a sufficient number of samples of the state, and the second will transform them into a probability distribution, the seeked tomograms, on a given auxiliary space.   

The notion of a $C^*$-dual tomographic pair will be introduced that will abstract both, the Marmo-Manko quantizer-dequantizer formalism and D'Ariano's quorum-quantum estimator settings for quantum tomography.   At the same time necessary and sufficient conditions for the reconstruction of the original state will be provided.    

A particularly interesting framework for the application of the theory is provided by dynamical systems, that is, representations of groups in $C^*$-algebras.  This setting extends in a natural way the group quantum tomography developed by D'Ariano \textit{et al} and, consequently, will be analyzed in depth showing the existence of a natural Radon transform associated with a reference or fiducial state.   This extension of the Radon transform unifies both, the sampling and the positive transformation needed to develop a tomographic picture of the theory, it will provide a natural extension to Eq. (\ref{eq:quantum_radon_heisenberg}), and could, consequently, be called a \textit{quantum Radon transform}.      

Moreover, a reconstruction theorem, relying on the theory of frames applied to square integrable representations of groups, will be proved.   The particular instance of finite and compact groups will be used to illustrate the theory, as well as the Weyl-Heisenberg group and its associated Wigner-Weyl tomographic picture.

The rest of the paper will be organized as follows.  Section \ref{section_C-algebras} will be devoted to summarize basic notations and constructions of the theory of $C^*$-algebras. In Section \ref{sec:tomographicC*}, the main notions of a tomographic theory in $C^*$-algebras will be discussed.  Section \ref{section_particular group} will be devoted to develop the theory of quantum tomography on $C^*$-algebras based on dynamical systems; there both the natural $C^*$-extension of the Radon transform and a reconstruction theorem will be discussed.   Finally, in Section \ref{sec:examples}, some particular situations and applications illustrating the theory will be exhibited.


\section{$C^{*}$--algebras and Quantum Tomography}\label{section_C-algebras}

Given a quantum mechanical systems with associated Hilbert space $\mathcal{H}$, the self-adjoint part of the algebra of bounded operators $\mathcal{B}(\mathcal{H})$ is usually considered as the algebra of (bounded) observables of the system, however, as von Neumann \cite{Ne30} pointedly realized, it is often necessary to consider more general algebras of observables.  In what follows we will assume that observables are elements of a $C^*$--algebra $\mathcal{A}$ or, more especifically, of a von Neumann algebra, that is, a concrete realization of a $C^*$-algebra as a closed subalgebra of the $C^*$-algebra of bounded operators of a complex separable Hilbert space.

\subsection{States, observables and $C^*$-algebras}
A  $*$--algebra $\mathcal{A}$ is a complex Banach algebra with a norm $\|\hspace{-0.02cm}\cdot\hspace{-0.02cm}\|$ and an \textit{involution} operation $^*$ satisfying:,

\begin{enumerate}

\item $(a^*)^*=a$,

\item $(ab)^*=b^*a^*$,

\item $(a+\lambda b)^*=a^*+\bar{\lambda} b^*$,
\end{enumerate}
for all $a,b\in\mathcal{A}$ and $\lambda\in\mathbb{C}$. A $C^*$--algebra (there is a huge literature on the subejct, we may cite for instance \cite{Pe79} and references therein)  is a $*$--algebra $\mathcal{A}$ such that 
$$
\|a^*a\|=\|a\|^2,\qquad\forall a\in\mathcal{A}\, .
$$
We will assume that the algebras considered here are unital in the sense that there exists an element $\mathbbm{1}$ such that $\mathbbm{1} a=a \mathbbm{1}$ for all $a\in\mathcal{A}$.  

Oservables of the system will correspond to self-adjoint elements, that is elements such that $a^*=a$. The subspace of all self-adjoint elements is denoted by $\mathcal{A}_{sa}$ and constitutes the Lie--Jordan Banach algebra of observables of the corresponding quantum system (see \cite{Fa13} for more details).

In particular, as indicated above, we can consider the $C^*$--algebra $\mathcal{B}(\mathcal{H})$ of all bounded operators on the Hilbert space $\mathcal{H}$ equipped with the operator norm and the involution defined by the adjoint operation, that is, $\textbf{A}^*=\textbf{A}^\dagger$ where $\textbf{A}^\dagger$ denotes the adjoint operator of $\textbf{A}\in\mathcal{B}(\mathcal{H})$.  Observables in such case correspond to self-adjoint operators on $\mathcal{H}$.

The states of the system are normalized positive functionals on $\mathcal{A}$, that is, linear maps $\rho:\mathcal{A}\rightarrow\mathbb{C}$ such that
\begin{equation}
\rho(\mathbbm{1})=1\, , \qquad \rho(a^*a)\geq 0\, ,\qquad \forall a\in\mathcal{A} \, .
\end{equation}

In the case in which $\mathcal{A}=\mathcal{B}(\mathcal{H})$, because of Gleason theorem \cite{Gl57}, states are in one-to-one correspondence with normalized non-negative Hermitean operators $\boldsymbol{\rho}$ acting on the Hilbert space $\mathcal{H}$, often called, as in Sect. \ref{sec:introduction},  density operators.  
The relation between states $\rho$ of the $C^*$--algebra $\mathcal{B}(\mathcal{H})$ and density operators $\boldsymbol{\rho}$ on the Hilbert space $\mathcal{H}$ is given by the formula:
\begin{equation}\label{expected_value_c_star}
\rho(\textbf{A})=\textrm{Tr}(\boldsymbol{\rho}\textbf{A}),\qquad\forall\textbf{A}\in\mathcal{B}(\mathcal{H}).
\end{equation}

The space of states of a given $C^*$--algebra $\mathcal{A}$ will be denoted by $\mathcal{S}(\mathcal{A})$ and it is a convex weak$^*$-compact subset of the topological dual $\mathcal{A}^\prime$ of $\mathcal{A}$ \cite{Al78}.

Notice that according to the physical interpretation of the $C^*$--algebra $\mathcal{A}$ as the the $C^*$-algebra of the Lie-Jordan algebra of observables of a given physical system, when the algebra is commutative, it will be describing a classical system whereas non-commutativity will correspond to ``genuine'' quantum systems.   A unital commutative $C^*$-algebra is the $C^*$-algebra of continuous functions $C(\Omega)$ on a compact space $\Omega$, its observables are continuous real functions and states are defined by probability Radon measures on $\Omega$.  Another instance of classical systems is found when $\Omega$ is a measure space, $\mu$ a probability measure and the von Neumann algebra associate to it is given by the space $L^\infty (\Omega, \mu)$ of essentially bounded functions on $\Omega$.  In this sense the discussion to follow will comprehend classical tomography and the classical Radon transform as well.

A state $\rho$ of the $C^*$--algebra $\mathcal{A}$ represents the state of the physical system under consideration and the number $\rho(a)$, for any given observable $a\in\mathcal{A}_{sa}$, is interpreted as the expected value of the observable $a$ measured in the state $\rho$, consequently, it is also denoted by:
\begin{equation}\label{mean_value_conclusion}
\langle a\rangle_\rho=\rho(a) \, .
\end{equation}
In this sense, Eq.\,\eqref{expected_value_c_star} represents the expected value of the observable described by the operator $\textbf{A}$ when the system is in the state given by the density operator $\boldsymbol{\rho}$.

Each self-adjoint element $a\in\mathcal{A}_{sa}$ defines a continuous affine function $\check{a}$ on the space of states $\mathcal{S}(\mathcal{A})$,
\begin{equation}
\check{a} (\rho) =\rho(a) \, .
\end{equation}
A theorem by Kadison \cite{Ka51} states that the correspondence $a\rightarrow\check{a}$ is an isometric isomorphism from the self-adjoint part of $\mathcal{A}$ onto the space of all continuous affine functions from $\mathcal{S}(\mathcal{A})$ into $\mathbb{R}$. Thus, the self-adjoint part of the algebra of observables $\mathcal{A}_{sa}$ can be recovered directly from the space of states and its complexification provides the whole algebra \cite{Fa13}.


\subsection{The GNS construction}\label{ps_GNS}

The picture of quantum mechanical systems proposed by P.A.M. Dirac \cite{Di81} favours the assignment of a Hilbert space to a given quantum mechanical system.    Moreover starting with the algebraic picture provided by a $C^*$-algebra $\mathcal{A}$, Dirac's Hilbert space is recovered by means of the GNS construction \cite{Ge43,Se47} named after Isra\"il M. Gel'fand, Mark A. Naimark and Irving E. Segal, provided that a state $\rho$ is given.

The Hilbert space $\mathcal{H}_\rho$ is constructed as the completion of the quotient inner product space $\mathcal{A}/\mathcal{J}_\rho$, where $\mathcal{J}_\rho=\left\{a\in\mathcal{A}|\,\rho(a^*a)=0\right\}$,
is the Gel'fand ideal of null elements for $\rho$, and the inner product  $\langle \cdot, \cdot \rangle$ is defined as:
\begin{equation}\label{eq:bracket_rho}
\langle[a],[b]\rangle_\rho=\rho(a^*b),\qquad a,b\in\mathcal{A} \, ,
\end{equation}
where $[a]$ denotes the class $a+\mathcal{J}_\rho$ in the quotient space $\mathcal{A}/\mathcal{J}_\rho$.  

Thus, given a state $\rho$ on a $C^*$--algebra $\mathcal{A}$, there is canonical representation $\pi_\rho$ of $\mathcal{A}$ in the $C^*$--algebra on bounded operators of a Hilbert space $\mathcal{H}_\rho$ defined as:
\phantomsection\label{pi_rep_ps}\begin{equation}\label{pi_rep}
\pi_\rho(a)[b]=[ab],\qquad\forall a,b\in\mathcal{A}.
\end{equation}

The GNS construction provides a cyclic representation $\pi_\rho$ of $\mathcal{A}$ with the cyclic vector corresponding to the unit element $\mathbbm{1}$.   Such vector $\pi_\rho (\mathbbm{1})$ will be called the \textit{vacuum vector} of $\mathcal{H}_\rho$ and denoted by $|0\rangle$. Moreover, we get that the state $\rho$ is also described by
\begin{equation}\label{eq:rhoa}
\rho(a)=\langle0|\pi_\rho(a)|0\rangle,\qquad a\in\mathcal{A} \, .
\end{equation}
In addition, given any element $a\in\mathcal{A}$, we have the associated vector $\pi_\rho(a)|0\rangle=[a \mathbbm{1}]=[a]$ that,  in what follows, will be denoted by  $|a\rangle \in\mathcal{H}_\rho$, thus
\begin{equation}
\pi_\rho(a)|0\rangle=|a\rangle.
\end{equation}
Conversely, each unitary vector $|a\rangle\in\mathcal{H}_\rho$ defines a state on $\mathcal{A}$ by means of
\begin{equation}\label{eq:vector_states}
\rho_a(b)=\langle a|\pi_\rho(b)|a\rangle=\rho(a^*ba) \, .
\end{equation}
Such states will be called \textit{vector states} of the representation $\pi_\rho$. More general states, called normal states, can be defined by means of density operators $\boldsymbol{\sigma}$ in $\mathcal{B}(\mathcal{H}_\rho)$ by the formula:
\begin{equation}\label{folium_rho}
\sigma(a)=\Tr\!\big(\boldsymbol{\sigma}\pi_\rho(a)\big),\qquad\forall a\in\mathcal{A} \, .
\end{equation}
Notice that, in such case,
\begin{equation}\label{folium}
\sigma(\mathbbm{1})=1,\qquad\sigma(a^*a)=\langle a|\boldsymbol{\sigma}|a\rangle\geq 0,\quad\forall a\in\mathcal{A} \, .
\end{equation}
The family of states given by \eqref{folium_rho} is called the \textit{folium} of the representation $\pi_\rho$ (see \cite{Ha96}, page 124).

The tomographic description of the state $\rho$ of a quantum system described by the $C^*$-algebra $\mathcal{A}$ will consist of assigning to this state a probability density $\mathcal{W}_\rho$ in some auxiliary space $\mathcal{N}$, in such a way that given $\mathcal{W}_\rho$ the state $\rho$ can be reconstructed unambiguously \cite{Ib09}, such reconstruction will be called the tomographic problem (see Fig.~\ref{Tom_problem_fig}).
\begin{figure}[h]
\centering
\includegraphics{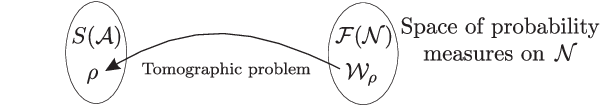}
\caption{Tomographic problem.}
\label{Tom_problem_fig}
\end{figure}

There is, however, not a single ``tomographic theory'', neither a canonical way to construct a tomogram $\mathcal{W}_\rho$ out of $\rho$.  In what follows, we will individuate the main ingredients that a tomographic description must contain: a \textit{Sampling Theory} and a \textit{Positive Transform}. We will discuss these two basic ingredients that together will constitute what we will call a \textit{Quantum Radon Transform}. 


\section{Tomographic description on $C^*$-algebras}\label{sec:tomographicC*}


\subsection{Sampling on $C^*$--algebras}\label{section_Sampling_C}\label{sec:sampling}

Let $\mathcal{A}$ be a $C^*$-algebra with unit $\mathbbm{1}$.  A tomographic theory begins by extracting information from the system, that is, by `sampling it'. This will be achieved by introducing the notion of a tomographic set.

\subsubsection{Tomographic sets: the sampling map}
Consider a family of elements $\{ U(x) \mid x \in \MM \}$ in $\mathcal{A}$ parametrized by an index $x\in \MM$ which can be discrete or continuous. This family can be described by a map $U\colon \MM\rightarrow\mathcal{A}$ where $\MM$ is a measurable space labelling the elements $U(x)$.
\begin{figure}[htbp]
\centering
\includegraphics{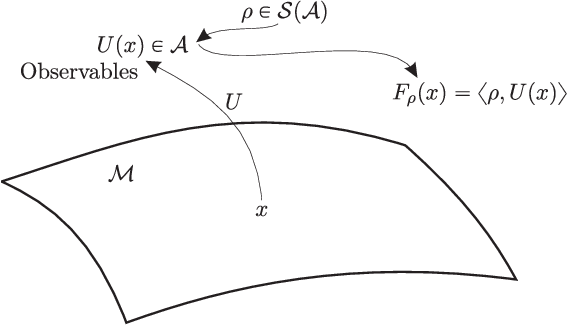}
\caption{Tomographic map $U$.}
\label{Tom_set}
\end{figure}

Given a state $\rho$ and a set $\{ U(x) \mid x \in \MM\}$, we will call the function $F_\rho:\mathcal{M}\rightarrow\mathbb{C}$ defined as:
\begin{equation}
F_\rho(x)=\langle\rho,U(x)\rangle=\rho\big(U(x)\big),\qquad x\in\MM,
\end{equation}
the \textit{sampling function} of $\rho$ with respect to $U$, Fig.~\ref{Tom_set}. In what follows, we will use indistinctly the notation $\rho(a)$ or $\langle\rho,a\rangle$ to denote the evaluation of the state $\rho$ on the element $a\in\mathcal{A}$.

We will say that the map $U$ \textit{separates states} if given two states $\rho$ and $\tilde{\rho}$ there exist $x\in\MM$ such that $F_\rho(x)\neq F_{\tilde{\rho}}(x)$.  

\begin{definition}  Let $\mathcal{A}$ be a $C^*$-algebra and $\MM$ a measurable space with measure $\mu$.  
A family of elements $\mathcal{U} = \left\{\vphantom{a^\dagger}U(x)\,|\,x\in\MM\right\}$ will be called a tomographic set and the map $U \colon \MM \to \mathcal{A}$, $x \mapsto U(x)$, a tomographic map, if it satisfies:
\begin{enumerate}
\item For every state $\rho \in \mathcal{S}(\mathcal{A})$, the sampling function $F_\rho (x ) = \langle \rho, U(x) \rangle$ is $\mu$-measurable and belongs to $L^2(\MM, \mu)$.
\item The map $U$ separates states.
\end{enumerate}
\end{definition}

Given a tomographic set $\mathcal{U}$, the map:
$$
\mathcal{F}_U \colon \mathcal{S}(\mathcal{A}) \to L^2(\MM, \mu) \, ,
$$ 
defined as $\mathcal{F}_U (\rho ) = F_\rho = \langle \rho , U(\cdot ) \rangle$, will be called the sampling map associated to 
$\mathcal{U}$. 

A number of comments are in order here.   In the definition of a tomographic set a different integrability condition could have been chosen, for instance, we could have asked $F_\rho$ to be integrable instead of square integrable.     The square integrability condition for the sampling functions will allow us to use Hilbert space techniques and they will prove to be critical when discussing one of the main families of examples, the so called group quantum tomography (see Sect. \ref{section_recosntruction_states_groups}).     

The notion of tomographic set introduced here is closer in spirit to the notion of frames used in the classical theory of sampling  \cite{Da92}, \cite{Ka94}.  There are various extensions of the classical theory of sampling to quantum systems (see for instance \cite{Fe15}, \cite{Ga19} and references therein). In this sense, notice that the tomographic map $U$ (and a dual tomographic map  $D$ discussed later on) is a far reaching generalization of the notion of frame (and its dual frame).  We will analyse this situation in detail in the case of group quantum tomography in Sect. \ref{section_recosntruction_states_groups}.   We will just point out here that there is a mature theory of sampling on $C^*$-algebras that extends the well-known results of standard sampling theory (see \cite{Fr02} and references therein) but we will not make use of them beyond the notions discussed in this section.

The notion of tomographic map $U$ has appeared in the literature under different names, for instance in the work by D'Ariano \textit{et al} is called a \textit{quorum} \cite{Da03}.


\subsubsection{Tomographic pairs: the reconstruction map}\label{sec:reconstruction}

A tomographic theory should allow for the reconstruction of the state of the system from its samples.  This will be achieved by introducing a map $D$ that is `dual' to the tomographic map $U$.
Thus, given a tomographic map $U$ we would require the existence of another map $D\colon \MM\rightarrow \mathcal{A}'$ 
which reconstructs the state $\rho$ from the sampling functions $F_\rho$, i.e., such that it will satisfy the reconstruction property:
\begin{equation}\label{reconstruction_state_D}
\rho = \int_{\MM} F_\rho (x) D(x) \diff \mu (x) \, ,
\end{equation}
for any state $\rho$.  The family $\mathcal{D} = \{ D(x) \mid x \in \MM\}$ will be called a dual tomographic set.

These considerations will lead us to the following definition:

\begin{definition}\label{biortho}  Let $\mathcal{A}$ be a  $C^*$-algebra.  A map $D \colon \MM \to \mathcal{A}'$ will be called a dual tomographic map of the tomographic map $U \colon \MM \to \mathcal{A}$ if it satisfies:
\begin{enumerate}
\item The function $|| D(x) || = \sup_{|| a || = 1} | \langle D(x), a \rangle |$, belongs to $L^2(\MM, \mu)$.
\item The function $\kappa(y,x) = \langle D(x),U(y)\rangle$, $x,y\in\MM$,
satisfies the reproducing condition:
\begin{equation}\label{eq:reproducing}
\Phi (y) = \hspace{-.26cm}\int\limits_{\hspace{.3cm}\MM}\hspace{-0.1cm}\kappa (y,x)\Phi(x)\diff\mu(x),
\end{equation}
for any $\Phi$ in the range of the sampling map $\mathcal{F}_U$.
\end{enumerate}
Under these conditions we will call the pair of maps $U$, $D$ a tomographic pair, the set $D(x)$ a \textit{dual tomographic set} to $U(x)$ and $\kappa$ will be called the kernel of the tomographic pair.  
\end{definition}

The reproducing equation (\ref{eq:reproducing}) can be understood  in the sense of distributions and the function $\kappa$ can be thought, for instance, to be the Dirac's delta distribution.  In this situation the maps $U,D$ are also said to be biorthogonal.    It is also common to call the functions $F_\rho(x)$ the ``tomographic symbols'' of the state $\rho$, see for instance \cite{As15_2}.

Note that in parallel with the definition of a tomographic set, the first condition in the definition of a dual tomographic set implies that the function $F_a(x) = \langle D(x), a \rangle$ is square integrable.  In fact:
$$
\int_{\MM} |F_a(x)|^2 \, \diff \mu (x) \leq \int_{\MM} ||D(x)||^2  ||a||^2 \diff \mu (x) = || D ||^2_{L^2(\MM)} ||a ||^2 \, .
$$

As in the case of the tomographic map $U$, there is a natural map $\mathcal{F}_D \colon L^2 (\MM, \mu) \to \mathcal{A}'$, defined as:
\begin{equation}\label{eq:reconstruction}
\mathcal{F}_D(\Phi ) = \hspace{-.26cm}\int\limits_{\hspace{.3cm}\MM}\hspace{-0.1cm}\Phi(x)D(x)\diff\mu(x) \, , \qquad \Phi \in L^2(\MM, \mu) \, ,
\end{equation}
Note that the element $\mathcal{F}_D(\Phi )$ is a continuous linear map on $\mathcal{A}$ because:
\begin{eqnarray}\label{eq:continuous}
|\langle \mathcal{F}_D(\Phi ) , a \rangle | &\leq&  \int_{\MM} |\Phi(x)| | F_a(x) |\diff\mu(x) \\ &\leq&  \int_{\MM} |\Phi(x)| \, || D(x) ||\, || a|| \diff\mu(x) \leq    || \Phi ||_{L^2(\MM)} || D ||_{L^2(\MM)} || a || \, , \nonumber
\end{eqnarray}
where the Cauchy-Schwarz inequality has been used in the r.h.s. of the previous equation.

We would like $\mathcal{F}_D$ to be a left inverse of $\mathcal{F}_U$, in which case we will properly call $\mathcal{F}_D$ a reconstruction map as it would allow to reconstruct the state $\rho$ from its sampling function $F_\rho$ by using Eq. (\ref{reconstruction_state_D}).


\subsubsection{Normalization}\label{sec:normalization}
 It is noticeable that the function $F_\rho (x) \langle D(x), a \rangle $ is integrable for all $a \in \mathcal{A}$, in fact:
 $$
 \int_{\mathcal{M}} | F_\rho (x) \langle D(x), a \rangle | \diff \mu (x) \leq || F_\rho ||_{L^2} || D(\cdot) ||_{L^2} || a || \, ,
 $$
 hence we may assume the convenient normalization condition:
 \begin{equation}\label{eq:normalization}
  \int_{\mathcal{M}}  F_\rho (x) \langle D(x), \mathbbm{1} \rangle\,  \diff \mu (x) = 1 \, .
 \end{equation}
 
 If condition (\ref{eq:normalization}) is satisfied we will say that $\widetilde{U}$, $\widetilde{D}$ is a normalized tomographic pair.  In what follows we will always assume this to be the case.


\subsubsection{The reconstruction theorem}

With the notations and assumptions introduced in the previous sections, we have the following theorem:

\begin{theorem}\label{GFT}
Let $U$, $D$ be a normalized tomographic pair on the $C^*$--algebra $\mathcal{A}$. Then the reconstruction map $\mathcal{F}_D  \colon L^2(\MM,\mu)\rightarrow\mathcal{A}^\prime$ is a left-inverse of the sampling map $\mathcal{F}_U \colon \mathcal{S}(\mathcal{A}) \rightarrow L^2(\MM,\mu)$.  In particular the reconstruction equation (\ref{reconstruction_state_D}) holds.
\end{theorem}

\begin{proof}
Let $F_\rho$ be the sampling function associated with the state $\rho$ by the sampling map $\mathcal{F}_U$, i.e., $\mathcal{F}_U(\rho ) = F_\rho$, then, we will show that $\mathcal{F}_D (F_\rho)$ is equal to $\rho$. First, we observe that the functional $\tilde{\rho} = \mathcal{F}_D (F_\rho)$ defined by:
\begin{equation*}
\tilde{\rho}(a)=\hspace{-.26cm}\int\limits_{\hspace{.3cm}\MM}\hspace{-0.1cm} F_\rho(x)\langle D(x),a\rangle\diff\mu(x),\qquad a\in\mathcal{A}\, ,
\end{equation*}
is continuous because of Eq. (\ref{eq:continuous}). 

Moreover, because the tomographic pair is normalized we have:
$$
\tilde{\rho} (\mathbbm{1})= \hspace{-.26cm}\int\limits_{\hspace{.3cm}\MM}\hspace{-0.1cm} F_\rho(x)\langle D(x),\mathbbm{1} \rangle\diff\mu(x)   = 1 \, .
$$

To show that $\rho=\tilde{\rho}$, we observe that:
\begin{eqnarray*}
\langle\tilde{\rho},U(y)\rangle &=& \hspace{-.26cm}\int\limits_{\hspace{.3cm}\MM}\hspace{-0.1cm} F_\rho(x)\langle D(x),U(y)\rangle\diff\mu( x ) \\ &=&   \hspace{-.26cm}\int\limits_{\hspace{.3cm}\MM}\hspace{-0.1cm} F_\rho(x) \kappa (y,x) \diff\mu(x) = F_\rho(y)=\langle\rho,U(y)\rangle \, ,
\end{eqnarray*}
but as $U$ separates states, then $\tilde{\rho} = \rho$.
\end{proof}

The formalism involving a tomographic map $U$ and a tomographic dual map $D$ has been widely used for applications in different settings and it was introduced by G. Marmo and V. Man'ko under the name of the ``quantizer-dequantizer'' formalism (see for instance \cite{Ma02}, \cite{Ci17}).


\subsubsection{Positivity: a simple instance of tomographic pairs}

Given a tomographic set $U$ it is not obvious how to construct a dual tomographic set $D$ for it.   The use of an auxilary state will help to do it.  Let $\rho_0$ a fixed fiducial state on $\mathcal{A}$.  Consider a map $U' \colon \MM \to \mathcal{A}$, then the map $D \colon \MM \to \mathcal{A}'$ defined by:
$$
\langle D(x), a \rangle = \langle \rho_0, U'(x)a \rangle \, ,  \qquad a \in \mathcal{A} \, , 
$$
could be used to try to construct dual tomographic sets.   In particular if the map $U'(x) = U(x)^*$ were used, then the kernel of the tomographic pair would be given as:
$$
\kappa (y,x) = \langle D(x), U(y) \rangle =\langle  \rho_0, U(x)^* U(y) \rangle \, , \qquad x,y \in \MM \, .
$$

In Sect. \ref{sec:dynamics} it will be shown that when the auxiliary space is a group $G$ and the tomographic map $U$ is a square integrable representation this construction actually provides a tomographic pair (with the appropriate normalization).

There is a natural notion of positivity associated with the construction of the sampling function and the kernel in this instance and that will be exploited later on.   We will say that a function $F \colon \MM\times\MM\rightarrow\mathbb{C}$ is of positive type, or positive semidefinite, if for all $N\in\mathbb{N}$, $\xi_i\in\mathbb{C}$ and any $x_i\in\MM$, $i=1,\ldots,N$, it satisfies that
\begin{equation}\label{positive_function}
\sum_{i,j=1}^N\bar{\xi}_i\xi_j F(x_i,x_j)\geq 0.
\end{equation}
Then we have the following lemma:

\begin{lemma}\label{tom_positive}
Given a state $\rho$ and a tomographic set $U\colon \MM\rightarrow\mathcal{A}$ on a $C^*$--algebra $\mathcal{A}$, the two-points sampling function $F_\rho(x,y)=\langle\rho,U(x)^*U(y)\rangle$, $x,y\in\MM$, is of positive type.
\end{lemma}

\begin{proof} It follows after a straightforward computation:
\begin{multline*}
\sum_{i,j=1}^N\bar{\xi}_i\xi_jF_\rho(x_i,x_j)=\hspace{-0.1cm}\sum_{i,j=1}^N\bar{\xi}_i\xi_j\langle\rho,U(x_i)^*U(x_j)\rangle\\
=\langle\rho,\hspace{-0.1cm}\sum_{i,j=1}^N\bar{\xi}_i\xi_jU(x_i)^*U(x_j)\rangle=\langle\rho,\left(\sum_{i=1}^N\xi_iU(x_i)\right)^{\hspace{-0.1cm}*}\hspace{-0.1cm}\left(\sum_{j=1}^N\xi_jU(x_j)\right)\rangle\geq 0.
\end{multline*}
\end{proof}

We will use this notion in establishing the properties of the second ingredient of a tomographic picture of states on a $C^*$-algebra.  Moreover, later on, we will take advantage of this property when dealing with group quantum tomography. 

Summarizing, we may say that the theory sketched in this section consists basically on reconstructing the positive form $\rho$ from a set of samples $F_\rho(x)$ by means of the generalized Fourier transform provided by a tomographic pair, Eq. (\ref{reconstruction_state_D}),  Fig.~\ref{Samp_diagr}.

As it was previously indicated there is a mature theory of frames on Hilbert $C^*$-modules and $C^*$-algebras \cite{Fr02} that could be used for such purpose, however, we will restrict ourselves to the simpler picture described in this section as it fits naturally into the many applications of quantum tomography commented in the introduction. 
\begin{figure}[h]
\centering
\includegraphics{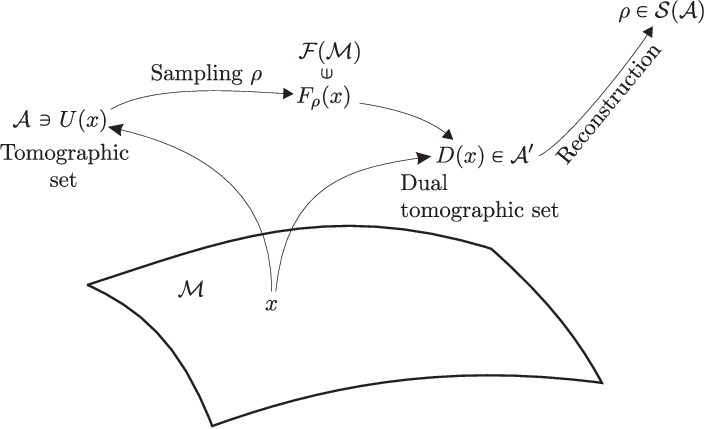}
\caption{Sampling diagram.}
\label{Samp_diagr}
\end{figure}


\subsection{A Positive Transform}\label{ps_GPT}\label{sec:positive}

The problem we are facing now is how to actually `measure' the samples $F_\rho(x)$.  Notice that in general it is not possible to determine the sampling function $F_\rho$ by direct measurements because it is a complex function and the quantities that can be measured are real probability distributions.  For that reason, we need a transformation that will allow us to obtain the sampling function $F_\rho$ from a marginal probability distribution $\mathcal{W}_\rho$.
Thus, the second tool in our tomographic programme for states on $C^*$-algebras is the choice of a \textit{Positive Transform} $\mathcal{R}$ that relates the sampling function $F_\rho$ to probability distributions.

One way to do that is by considering a second auxiliary space $\NN$ that parametrizes a family of distributions $R(y)$ on $\MM$, $y\in\NN$, and define the transform of a function $F$ on $\MM$ as the function $\mathcal{R}(F) \colon \NN \to \mathbb{C}$  given by:
\begin{equation}\label{eq:RF}
\mathcal{R}(F)(y)=\langle R(y),F\rangle \, ,
\end{equation}
where $\langle\cdot,\cdot\rangle$ denotes the natural pairing between of the distribution $R(y)$ and the function $F$.  Again, if $\NN$ is a measure space with measure $\sigma$, we will assume that $R$ is integrable in the weak sense, that is, the function $\mathcal{R}(F)=\langle R(\cdot),F\rangle$ is $\sigma$--integrable for any $F$ integrable.   If we assume in addition that the function $\langle R(y),F\rangle$ is in $L^\infty(\NN, \sigma)$, then we get:
$$
|| \mathcal{R}(F) ||_{L^1(\NN)} \leq \int_{\NN} |\langle R(y),F\rangle | \diff \sigma (y) \leq K \| F \|_\infty \, ,
$$
and $\mathcal{R} \colon C_b(\MM ) \to L^1(\NN , \sigma)$ is continuous.

This is what is done for the Classical Radon Transform.  The $C^*$-algebra would be, for instance, the space of continuous functions on a compact subset $\Omega \in \mathbb{R}^2$.  Then
$\NN$ will be the set of lines $l$ on $\mathbb{R}^2$ and $R$ will be the map $R \colon \NN\rightarrow C(\Omega)^\prime$, 
$R \colon l \rightsquigarrow\delta_l$, where $\delta_l$ is the delta distribution along $l$,
and the tomogram will be (recall  Eq. \eqref{Radon_Transform}):
\begin{equation}\label{eq:Wrho_again}
\mathcal{W}_\rho(l)=\langle R(l), F_\rho \rangle = \hspace{-0.1cm}\int\limits_{\hspace{.3cm}l}\hspace{-0.1cm}F_\rho \big(q(s),p(s)\big)\diff s \, ,
\end{equation}
with $F_\rho (q,p )$ a function describing a state $\rho$ on $C(\Omega)$, i.e., a Radon measure.

Thus, in general, we would like to identify the image $\mathcal{R}(F)$ of the function $F$ as a density probability or, as the Radon-Nikodym derivative of an absolutely continuous measure $\sigma_F$, with respect to the measure $\sigma$, that is, we would like to determine the class of functions $F$ whose transform will be positive.  
For instance, if $\mathcal{N} = \mathcal{M} = \mathbb{R}$, because of Bochner's theorem, the inverse Fourier transform of a Borel probability measure $\sigma$ is a positive definite function $F (x)  = \frac{1}{2\pi} \int  e^{ixy} \diff \sigma (y)$ on $\mathbb{R} = \MM$.  
We have already realized that the function $F_\rho(x,y) = \langle \rho, U(x)^* U(y) \rangle$ associated to a tomographic set $U(x)$ is positive definite (Lemma\,\ref{tom_positive}) and we will use in in the coming section to construct a positive transform.

Then we will say that the map $\mathcal{R}\colon F\rightsquigarrow\mathcal{R}(F)$ is a \textit{Positive Transform} if it maps functions of the form $F_\rho$ on $\MM$ into non-negative functions on $\NN$, i.e., if $F\colon \MM \rightarrow\mathbb{C}$ is in the range of $\mathcal{F}_U$, then:
\begin{equation}
\mathcal{R}(F)(y)=\langle R(y),F\rangle\geq 0,\qquad\forall y\in\NN.
\end{equation}

We will say that $\mathcal{R}$ is non-degenerate if it has a left inverse.
Under this rather long list of conditions, we conclude by noticing that if $\rho$ is a state, chosings the appropriate normalizations, $\mathcal{R}(F_\rho)$ will be a normalized non-negative function on $\NN$, that is:
\begin{equation}
\hspace{-.26cm}\int\limits_{\hspace{.3cm}\NN}\hspace{-0.1cm}\mathcal{R}(F_\rho)(y)\diff\sigma(y)=1.
\end{equation}
Moreover, if we know $\mathcal{R}(F_\rho)$, we could obtain $F_\rho$ by applying a left-inverse map $\mathcal{R}^{-1}$, i.e., $F_\rho=\mathcal{R}^{-1}\circ\mathcal{R}(F_\rho)$. The function $\mathcal{R}(F_\rho)$ will be called the \textit{tomogram} of the state $\rho$ and we will denote it by $\mathcal{W}_\rho$ (see Fig.~\ref{Pos_diagr}):
\begin{equation}\label{tomogram_conclusion}
\mathcal{W}_\rho(y)=\langle R(y),F_\rho\rangle.
\end{equation}
Notice again that the tomogram $\mathcal{W}_\rho(y)$ satisfies that it is a probability density related with the state $\rho$:
\begin{equation}
\mathcal{W}_\rho\geq 0,\qquad\hspace{0cm}\int\limits_{\hspace{.3cm}\NN}\hspace{-0.1cm}\mathcal{W}_\rho(y)\diff\sigma(y)=1.
\end{equation}

\begin{figure}[h]
\centering
\includegraphics{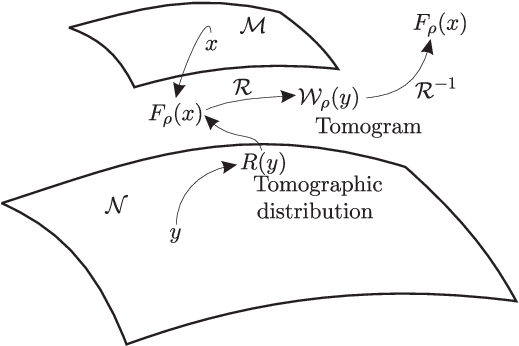}
\caption{Positive Transform diagram.}
\label{Pos_diagr}
\end{figure}


\section{Quantum Tomography and groups}\label{section_particular group}

Given a state $\rho$ on a $C^*$-algebra $\mathcal{A}$, it is not apparent how to construct the two ingredients, a sampling theory and a positive transform, needed for its tomographic description.   There is, however, a situation when both ingredients are interlocked and arise in a natural way.  This happens when there is a group $G$ represented on the system, sometimes called a dynamical system.  
Such situation occurs, for example, in Spin Tomography \cite{Ma97} with the group $SU(N)$ (see Sect. \ref{section_Spin_Tomography}) or, in standard Quantum Tomography, with the Weyl-Heisenberg group (Sect. \ref{resume_heis}).

In this section, a tomographic theory, following the steps laid in Sect. \ref{sec:tomographicC*}, for a class of normal states in dynamical systems on $C^*$-algebras will be described.

\subsection{Equivariant tomographic theories on $C^*$--algebras}\label{section_equivariant}

In many situations of interest, either describing symmetries of the system or the dynamical background of the theory, there is a group present in the system whose states we want to describe tomographically.
We will assume that there is a Lie group $G$ acting on the $C^*$--algebra $\mathcal{A}$, that is, there is a strongly continuous map $T\colon G\rightarrow\textrm{Aut}(\mathcal{A})$, such that:
\begin{equation}
T_e=\mathbbm{1},\qquad T_{g_1}T_{g_2}=T_{g_1g_2},\qquad\forall g_1,g_2\in G.
\end{equation}
The triple $(G,T,\mathcal{A})$ is also called a dynamical system (there is a canonical construction of a $C^*$-algebra, the crossed product algebra, associated with this structure that will not be needed in what follows).

If we have a tomographic theory for the states of $\mathcal{A}$, that is a tomographic pair $U\colon \MM \to \mathcal{A}$, $D\colon \MM \to \mathcal{A}'$ and a positive transform $R \colon \NN \to \mathcal{D}(\MM)$, then we will assume, in order to have a consistent theory, that the group $G$ also acts on the auxiliary spaces $\MM$ and $\NN$.   Such actions will be simply denoted by $x\rightsquigarrow g\cdot x$ and $y\rightsquigarrow g\cdot y$, $g\in G$, $x\in\MM$ and $y\in\NN$, respectively. Then it will be assumed that both maps $U$, $D$ are equivariant, that is:
\begin{equation}\label{equivariant_U}
T_g U(x) = U(g\cdot x)\, , \qquad\forall x\in\MM\quad\mbox{and}\quad\forall g\in G \, ,
\end{equation}
and
 \begin{equation}\label{equivariant_D}
T_{g^{-1}}^* D(x) = D(g\cdot x)\, ,\qquad\forall x\in\MM \quad\mbox{and}\quad\forall g\in G \, ,
\end{equation}
with $T_g^*$ the natural action of $G$ on the space of states, that is $\langle T_g^* \rho, a \rangle = \langle \rho, T_g a\rangle$, $\forall a \in \mathcal{A}, \rho \in \mathcal{S}(\mathcal{A})$.

Under these conditions, it is easy to conclude that the sampling function $F_\rho$ verifies:
\begin{equation}\label{property_equivariant_rho}
F_\rho(g\cdot x)= F_{g^*\rho}(x)\, , \qquad \forall x \in \MM \, , \quad g\in G\, ,
\end{equation}
since
\begin{equation}
F_\rho(g\cdot x) = \langle\rho,U(g\cdot x)\rangle = \langle T_g^*\rho, U(x)\rangle = F_{g^*\rho}(x),
\end{equation}
where $g^*\rho=T_g^*\rho$. Notice that if $\rho$ is an invariant state, $T_g^*\rho=\rho$, then, the corresponding sampling function will be invariant:
\begin{equation}
F_\rho(g\cdot x)=  F_\rho(x),\qquad\forall g\in G,\quad x\in\MM.
\end{equation}

As indicated before, we will also consider that the group $G$ acts on the auxiliary space $\NN$ used to define the Positive Transform. If we assume that the map $R:\NN\rightarrow\mathcal{D}(\MM)$ is equivariant, i.e.:
\begin{equation}
g_*R(y) = R(g^{-1}\cdot y) \, ,
\end{equation}
where $g_*$ indicates now the natural action induced on the space of distributions $\mathcal{D}(\MM)\subset\mathcal{F}(\MM)^\prime$ given by the action of $G$ on $\MM$, more explicitly:
\begin{equation}
\langle g_*R(y),F\rangle=\langle R(y),g^*F\rangle\quad\mbox{and}\quad g^*F(x)=F(g\cdot x).
\end{equation}
If $R$ is actually a positive transform and $\mathcal{W}_\rho$ the tomogram of the state $\rho$, we will get:
\begin{eqnarray}\label{eq:W_equivariance}
\mathcal{W}_\rho(g^{-1}\cdot y) &=& \langle R(g^{-1}\cdot y),F_\rho\rangle = \langle g_* R(y),F_\rho\rangle
\\ &=& \langle R(y),g^* F_\rho\rangle = \langle R(y),F_{g^*\rho}\rangle=\mathcal{W}_{g^*\rho}(y) \, , \nonumber
\end{eqnarray}
and, in the particular instance that  $\rho$ is an invariant state, we conclude by observing that its tomogram will be invariant too:
\begin{equation}
\mathcal{W}_\rho(g\cdot y)=\mathcal{W}_\rho(y),\qquad\forall g\in G \, .
\end{equation}


\subsection{Quantum tomography and dynamical systems}\label{sec:dynamics}

Thus, if we want to construct a tomographic theory using  a group $G$, it is natural to consider auxiliary spaces $\MM$ and $\NN$ which are homogeneous spaces for the group, in particular we can consider the group itself.  Hence, we will concentrate in the particular instance that the auxiliary spaces $\MM$ and $\NN$ are a Lie group $G$.  

Moreover we will assume that the group $G$ is represented by inner autormorphisms of the $C^*$-algebra $\mathcal{A}$, that is, we will assume that the tomographic map $U \colon G\rightarrow\mathcal{A}$ is provided by a  strongly continuous unitary representation $U$ of $G$ on $\mathcal{A}$, thus for all $g,g_1,g_2 \in G$:
\begin{equation}
U(g_1g_2)=U(g_1)U(g_2)\, ,\quad U(e)=\mathbbm{1}\,, \quad \text{and}\quad U(g)^* = U(g)^{-1}=U\left(g^{-1}\right) \, ,
\end{equation}
and the action of $G$ on $\mathcal{A}$ is given by the group homomorphism $T_g:G\rightarrow\textrm{Aut}(\mathcal{A})$:
\begin{equation}
T_g(a)=U(g) a\,U(g)^*,
\end{equation}
with $a\in\mathcal{A}$ and $g\in G$.  We see immediately that
\begin{equation}
U(ghg^{-1})=U(g) U(h)\hspace{0.02cm}U(g)^* =T_g(U(h)) ,\qquad\forall g,h\in G,
\end{equation}
which is the equivariant property \eqref{equivariant_U} for the conjugation action of $G$ on itself.

The sampling function corresponding to the state $\rho$ is given by:
\begin{equation}
F_\rho(g)=\langle\rho,U(g)\rangle \, ,
\end{equation}
and the function $F_\rho \colon G\rightarrow\mathbb{C}$ is of positive type because the map  $F_\rho(g,h) = \langle\rho,U(g)^* U(h)\rangle = F_\rho(g^{-1}h)$ is of positive type, Lemma \ref{tom_positive}, then:
\begin{equation}\label{positivity_functions_2}
\sum_{i,j=1}^N\bar{\xi}_i\xi_jF_\rho(g_i^{-1}g_j)\geq 0,
\end{equation}
for all $N\in\mathbb{N}$, $\xi_i\in\mathbb{C}$, $g_i\in G$ with $i=1,\ldots,N$. Moreover a simple computation shows that $F_\rho$ satisfies \eqref{property_equivariant_rho} too.

If the $C^*$-algebra $\mathcal{A}$ is the $C^*$-algebra of all bounded operators on a Hilbert space $\mathcal{B}(\mathcal{H})$, because of the one-to-one correspondence between states and density operators, the sampling function $F_\rho$ can be written as
\begin{equation}\label{smeared_character_chap_2}
F_\rho(g)=\Tr\!\big(\boldsymbol{\rho}\,U(g)\big),
\end{equation}
for some density operator $\boldsymbol{\rho}$,
then, since the character of a finite-dimensional representation  $U$ of the group $G$ is defined as:
\begin{equation}\label{character_group}
\chi(g)=\Tr\!\big(U(g)\big)\, ,
\end{equation}
we will call the sampling function $F_\rho$ the \textit{smeared character} of the representation $U$ with respect to the state $\rho$ and, consequently, we will denote it in what follows by 
$$
\chi_\rho(g)= F_\rho(g) = \langle \rho, U(g) \rangle \, .
$$   
Note that if $\mathcal{H}$ is finite-dimensional with dimension $n$ and the state $\boldsymbol{\rho}$ is the normalized unit $\frac{1}{n}\mathbbm{1}$, its smeared character is just the (normalized) standard character of the group representation \eqref{character_group}.

Given a state $\rho$, the GNS construction described before, Sect. \ref{ps_GNS}, provides, a representation $\pi_\rho$ of $\mathcal{A}$ in $\mathcal{H}_\rho$, hence we get a strongly continuous unitary representation $U_\rho$ of $G$ given by:
\begin{equation}\label{U_rho_GNS}
U_\rho(g)=\pi_\rho\big(U(g)\big).
\end{equation}
Notice that $U_\rho(g)$ is actually a unitary operator on the Hilbert space $\mathcal{H}_\rho$ because, recall Eq. (\ref{eq:bracket_rho}):
\begin{equation}
\langle U_\rho(g)[a],U_\rho(g)[b]\rangle_\rho=\rho\big((U(g)a)^*(U(g)b)\big)=\rho(a^*b)=\langle[a],[b]\rangle_\rho,
\end{equation}
for all $g\in G$, $[a],[b]\in\mathcal{H}_\rho$.

Now, because of (\ref{eq:rhoa}), the sampling function of a representation $U$ corresponding to a state $\rho$ can be written as
\begin{equation}
\chi_\rho(g)=\langle 0|U_\rho(g)|0\rangle,
\end{equation}
where $|0\rangle$ is the fundamental vector of $\mathcal{H}_\rho$.  Moreover, fixed the state $\rho$, the smeared character of $U$ with respect to any other state $\sigma$ in the folium of $\rho$, Eq.\,\eqref{folium_rho}, will be given by
\begin{equation}\label{character_folium}
\chi_\sigma(g)=\langle\sigma,U(g)\rangle=\Tr\!\left(\boldsymbol{\sigma}\pi_\rho\big(U(g)\big)\vphantom{\text{\larger[3]$\boldsymbol{\mathbbm{1}}$}}\!\right)=\Tr\!\big(\boldsymbol{\sigma}U_\rho(g)\big).
\end{equation}

If $G$ is represented in the $C^*$-algebra $\mathcal{A}$, then the map $U \colon G \to \mathcal{A}$ extends to a $C^*$-algebra homomorphism, denoted as $\pi_U \colon C^*(G) \to \mathcal{A}$ between the envelopping $C^*$ algebra of the group $G$ and $\mathcal{A}$.    Given a state $\rho$ on $\mathcal{A}$, it induces a state $\rho_U$ on $C^*(G)$ as $\rho_U(a) = \rho (\pi_U(a))$ with $a \in C^*(G)$.    States are characterized by means of their smeared characters (or sampling functions):

\begin{theorem}\label{prop_positive}
Let $\rho\,\colon \mathcal{A} \to \mathbb{C}$ be a continuous linear function and consider the smeared character $\chi_\rho (g)=\langle\rho,U(g)\rangle$ where $U$ is a strongly continuous unitary representation of the Lie group $G$ on $\mathcal{A}$.  Then, $\rho$ is a state iff $\chi_\rho = F_\rho$ is a positive type function on $G$ and $\chi_\rho (e) = 1$.
\end{theorem}

\begin{proof} We have seen already, Eq. \eqref{positivity_functions_2}, that $\chi_\rho$ is of positive type if $\rho$ is a state and $\chi_\rho(e)=1$ because of the normalization of $\rho$. 

Conversely, if $\chi_\rho$ is a positive type function on $G$, because of Naimark reconstruction theorem \cite{Na64}, there exist a complex separable Hilbert space $\mathcal{H}$ supporting a strongly continuous unitary representation $V$ of $G$ and a vector $|0\rangle\in\mathcal{H}_\rho$ such that:
$$
\chi_\rho(g)=\langle 0|V_\rho(g)| 0\rangle.
$$
Actually such representation is obtained as the GNS representation of the $C^*$-algebra of the group $G$ determined by the state $\rho_0 \colon C^*(G) \to \mathbb{C}$ defined by $\rho_0(a) = \langle \rho, \pi_U(a) \rangle$.  Note that $\rho_0$ is a state because by hyphotesis:
$$
\rho_0(\mathbbm{1}) = \langle \rho, \pi_U(\mathbbm{1}) \rangle =  \langle \rho, \mathbbm{1}_\mathcal{A}\rangle =  \langle \rho, U(e)) \rangle  = \chi_\rho (e) = 1\, .
$$  
Moreover if $\chi_\rho$ is of positive type on $G$, then the natural extension of $\chi_\rho$ to $C^*(G)$ is positive. Then
$$
\rho_0(a^*a) = \langle \rho, \pi_U(a^*a) \rangle =   \chi_\rho (a^*a) \geq 0 \, .
$$ 
because $\chi_\rho$ is of positive type.  

Hence, the representation $\pi_{\rho_0} \colon C^*(G) \to \mathcal{B}(\mathcal{H}_{\rho_0})$ determines the representation $V$ of the group $G$ above.

Then because of Eq. \eqref{eq:rhoa} we conclude:
\begin{equation}\label{eq:rho0rho}
\langle \rho, \pi_U(a) \rangle = F_\rho (a) = \rho_0(a) = \langle 0 |\pi_{\rho_0}(a) | 0 \rangle \, ,
\end{equation}
Let $a \in C^*(G)$ be such that $a \in \ker \pi_U$, i.e., $\pi_U(a) = 0$, then because of \eqref{eq:rho0rho}, $\rho_0(a) = 0$. But $\ker \pi_U$ is a closed bilateral ideal and $C^*(G)/\ker \pi_U \cong \mathcal{A}$, then $\rho_0$ projects to a state on the quotient algebra  $C^*(G)/\textrm{ker}\, \pi_U$, but because of \eqref{eq:rho0rho} again, $\rho$ coincides with the projection of $\rho_0$ on $\mathcal{A}$ and $\rho$ is a state.
\end{proof}

Note that the state $\rho_0$ on $C^*(G)$ used in the previous proof is just the extension of the state $\rho$ on $\mathcal{A}$ to $C^*(G)$ by the homomorphism $\pi_U\colon C^*(G) \to \mathcal{A}$.  

In what follows we will assume that a fiducial state $\rho$ has been chosen and we will describe the tomographic theory of states $\boldsymbol{\sigma}$ on its folium. 


\subsection{The Quantum Radon Transform associated with a dynamical system}\label{section_tomograms_group}

Given the Lie group $G$ we can consider the space $\mathfrak{g}\times\mathbb{R}$, where $\mathfrak{g}$ denotes the Lie algebra of $G$, and the extended exponential map: $\exp\colon \mathfrak{g}\times\mathbb{R}\rightarrow G$, given by $\exp(\xi,t) = \exp(t\xi)$,  $t \in \mathbb{R}$, $\xi \in \mathfrak{g}$, where $\exp \colon \mathfrak{g}\rightarrow G$, is the ordinary exponential map.  In particular, if $G$ is a matrix Lie group, then,
\begin{equation}
\exp(t\xi)=\sum_{n=0}^\infty\frac{t^n}{n!}\xi^n \, .
\end{equation}
with $\xi$ a matrix in the matrix Lie algebra of $G$.

Given the representation $U$ of the group $G$ on the algebra $\mathcal{A}$ and the fiducial state $\rho$, we have the canonical unitary representation $U_\rho$ of $G$ on the Hilbert space $\mathcal{H}_\rho$ associated with the GNS construction, Sect. \ref{ps_GNS}. 
Also, given $\xi \in \mathfrak{g}$, we can consider the strongly continuous one-parameter group $U_t$ of unitary operators on $\mathcal{H}_\rho$ given by:
$$
U_t = U_\rho\big(\exp(t\xi)\big) = \pi_\rho(U(\exp(t\xi))\, ,
$$
with $t\in\mathbb{R}$.  Note that,
\begin{eqnarray*}
U_t^\dagger =  \pi_\rho(U(\exp(t\xi)^*) &=& \pi_\rho(U(\exp(t\xi)^{-1}) \\ &=& \pi_\rho(U(\exp(-t\xi))  = U_{-t} = U_t^{-1} \, .
\end{eqnarray*}

 Because of Stone's theorem \cite{St32}, there exists a densely defined self-adjoint operator $\boldsymbol{\xi}$ on $\mathcal{H}_\rho$ such that
\begin{equation}\label{selfadjoint_exp}
\e^{it\boldsymbol{\xi}}=U_\rho\big(\exp(t\xi)\big).
\end{equation}
Note that the element $\xi$ in the Lie algebra and the operator $\boldsymbol{\xi}$ have opposite symmetry because of the $i$ factor in the exponent in the r.h.s. of (\ref{selfadjoint_exp}), that is, if $G$ is a matrix Lie group, then $\xi\in\mathfrak{g}$ is a skew-Hermitian matrix while  the operator $\boldsymbol{\xi}$ is Hermitean.

Let us denote by $\Theta$ the canonical left-invariant Cartan $1$-form on $G$, i.e., the tautological $\mathfrak{g}$-valued 1-form given by $\Theta(\xi)=\xi$, where $\xi$ denotes both an element on the Lie algebra of $G$ and the corresponding left-invariant vector field on $G$.  

Let $\boldsymbol{\Theta}_\rho$ be the ``quantization'' of the Cartan $1$-form, i.e., $\boldsymbol{\Theta}_\rho$ is a left-invariant 1-form on $G$ with values in self-adjoint operators on $\mathcal{H}_\rho$ defined by:
$$
\boldsymbol{\Theta}_\rho (\xi )=\boldsymbol{\xi}\, ,\qquad  \forall\xi\in\mathfrak{g}\, . 
$$ 
In other words, $\boldsymbol{\Theta}_\rho$ maps the Lie algebra $\mathfrak{g}$  of the Lie group $G$ in the space of densely defined self-adjoint operators on $\mathcal{H}_\rho$, by assigning to each element $\xi \in \mathfrak{g}$ the generator $\boldsymbol{\xi}$ of the one-parameter group $U_t = \pi_\rho (U(\exp (t \xi))$ of unitary operators defined by $\exp t\xi$. 

We will often use the notation $\langle\boldsymbol{\Theta}_\rho,\xi\rangle=\boldsymbol{\xi}$ instead for the evaluation of the 1-form $\boldsymbol{\Theta}_\rho$ on the Lie algebra element $\xi$.  Thus if $E_a$, $a = 1, \ldots, r$, is a basis for the Lie algebra $\mathfrak{g}$, any element $\xi$ can be written as $\xi = \xi^a E_a$  (Einstein's summation convention understood). Then $\langle\boldsymbol{\Theta}_\rho,\xi\rangle= \boldsymbol{\Theta}_\rho (\xi )=\boldsymbol{\xi} = \xi^a \mathbf{E}_a$.  We will call $\boldsymbol{\Theta}_\rho$ the quantum Cartan 1-form associated with the representation $U$ and the state $\rho$ and in what follows we will omit the subindex $\rho$ if there is no risk of confussion.

It is easy to check that the operators $\boldsymbol{\xi}$ provide a representation of $\mathfrak{g}$ in $\mathcal{H}_\rho$, that is, it is satisfied:
\begin{equation}
[\boldsymbol{\xi},\boldsymbol{\zeta}]=i\langle\boldsymbol{\Theta},[\xi,\zeta]\rangle,\qquad\forall\xi,\zeta\in\mathfrak{g}.
\end{equation}

We may use the spectral theorem, \cite[Chap. 7]{Re80}, to represent each self-adjoint operator $\boldsymbol{\xi}$ on $\mathcal{H}_\rho$ as:
\begin{equation}\label{spectral_thm_ch_2}
\boldsymbol{\xi}=\hspace{-.25cm}\int\limits_{\hspace{-.04cm}-\infty}^{\hspace{.45cm}\infty}\hspace{-0.15cm}\lambda\, E_{\xi}(\diff\lambda),
\end{equation}
where $E_{\xi}$ denotes the spectral measure associated with $\boldsymbol{\xi}$. Then using \eqref{selfadjoint_exp}, we can write:
\begin{equation}\label{spectral_measure}
U_\rho\big(\exp(t\xi)\big)=e^{it\boldsymbol{\xi}}=\hspace{-.25cm}\int\limits_{\hspace{-.04cm}-\infty}^{\hspace{.45cm}\infty}\hspace{-0.15cm}\e^{it\lambda} E_{\xi}(\diff\lambda).
\end{equation}

Let $\sigma$ be a state on the folium of $\rho$, i.e., $\boldsymbol{\sigma}$ is a density operator on $\mathcal{H}_\rho$ and $\sigma$ is defined by Eq.\,\eqref{folium_rho}, then consider the measure $\mu_{\sigma,\xi}(\diff\lambda)=\textrm{Tr}\big(\boldsymbol{\sigma}E_{\xi}(\diff\lambda)\big)$, in other words, if $\Delta$ is a Borel set in $\mathbb{R}$, the probability $P(\boldsymbol{\xi},\sigma;\Delta)$ that the output of measuring the observable $\boldsymbol{\xi}$ will be in the set $\Delta$ when the system is in the state $\sigma$ is given by:
\begin{equation}\label{physical_int_measure}
P(\boldsymbol{\xi};\sigma ,\Delta)=\hspace{-0.15cm}\int\limits_{\hspace{.3cm}\Delta}\hspace{-0.1cm}\mu_{\sigma,\xi}(\diff\lambda)=\mu_{\sigma,\xi}(\Delta)=\textrm{Tr}\big(\boldsymbol{\sigma}E_{\xi}(\Delta)\big) \, .
\end{equation}
Then, obviously, we get $\mu_{\sigma,\xi}(\mathbb{R})=1$. Moreover, if the measure $\mu_{\sigma,\xi}(\diff\lambda)$ is absolutely continuous with respect to the Lebesgue measure $\diff X$ on $\mathbb{R}$, there will exist a function $\mathcal{W}_\sigma(X;\xi)$ in $L^1(\mathbb{R},\diff X)$ such that for all measurable $\Delta$:
\begin{equation}\label{absolutely_tom}
\int\limits_{\hspace{.3cm}\Delta}\hspace{-0.1cm}\mu_{\sigma,\xi}(\diff\lambda)=\hspace{-0.15cm}\int\limits_{\hspace{.3cm}\Delta}\hspace{-0.1cm}\mathcal{W}_\sigma(X;\xi)\, \diff X\geq 0.
\end{equation}

In general, this will not be true if the measure $\mu_{\sigma,\xi}(\diff\lambda)$ has singular part, for instance, if $\boldsymbol{\xi}$ has discrete spectrum, however, we get a similar result in such case too (see later on, Sect. \ref{section_recosntruction_states_groups}).

\begin{definition}\label{Radon_Nik}
Given a state $\boldsymbol{\sigma}$ in the folium of a reference state $\rho$ and a unitary representation $U$ of a Lie group $G$ on the unital $C^*$--algebra $\mathcal{A}$, the family of Borelian probability measures $\mu_{\sigma,\xi}(\diff\lambda)=\Tr\big(\boldsymbol{\sigma}E_{\xi}(\diff\lambda)\big)$, $\xi\in\mathfrak{g}$, will be called the quantum tomogram of $\boldsymbol{\sigma}$. 
\end{definition}

Then, the absolutely continuous part of the measure $\mu_{\sigma,\xi}(\diff\lambda)$ defines a function $\mathcal{W}_\sigma:\mathfrak{g}\times\mathbb{R}\rightarrow\mathbb{R}$ given by Eq.\,\eqref{absolutely_tom}, that will be also called the quantum continuous tomogram of $\boldsymbol{\sigma}$, in other words, $\mathcal{W}_\sigma(X;\xi)$ is the Radon--Nikodym derivative of the measure $\mu_{\sigma,\xi}(\diff X)$ with respect to the Lebesgue measure $\diff X$:
\begin{equation*}
\mathcal{W}_\sigma(X;\xi)=\frac{\delta\mu_{\sigma,\xi}(\diff X)}{\delta X}\, .
\end{equation*}

Notice that this  formula is another way of rewriting \eqref{absolutely_tom}.   Notice that if $\mathcal{W}_\sigma$ is continuous it is non-negative: $\mathcal{W}_\sigma\geq 0$.

From \eqref{character_folium} and \eqref{spectral_measure}, we get immediately:
\begin{equation}\label{smeared_character_Fourier}
\chi_\sigma\big(\hspace{-0.04cm}\exp(t\xi)\big)=\hspace{-.25cm}\int\limits_{\hspace{-.04cm}-\infty}^{\hspace{.45cm}\infty}\hspace{-0.15cm}\e^{itX}\mu_{\sigma,\xi}(\diff X),
\end{equation}
i.e., $\chi_\sigma\big(\exp(t\xi)\big)$ is the Inverse Fourier Transform of the measure $\mu_{\sigma,\xi}(\diff X)$ then, as we already know and, in agreement with Bochner's theorem, the smeared character (sampling function) $\chi_\rho$ is a function of positive type.  This, as indicated in Sect. \ref{sec:positive} (see the comments after Eq. (\ref{eq:Wrho_again})), will give us the clue to define the positive transform needed for a tomographic descripton of quantum states.  
The auxiliary space $\NN$ will be the set $\mathbb{R} \times \mathfrak{g}$, and the map $R \colon \NN \to \mathcal{D}(G)$ is the distribution defined by the Fourier kernel, that is: 
$$
R(X, \xi) = \frac{1}{2\pi}e^{-itX}\, ,
$$ 
i.e., the Positive Transform $\mathcal{R}(F)$ is given by, recall Eq. (\ref{eq:RF}):
$$
\mathcal{R}(F ) (X,\xi) = \langle R(X,\xi) , F\rangle = \frac{1}{2\pi}\hspace{-.15cm}\int\limits_{\hspace{-.04cm}-\infty}^{\hspace{.45cm}\infty}\hspace{-0.15cm}\e^{-itX} F(t,\xi) \, \diff t \, .
$$
 If the measure  $\mu_{\sigma,\xi}(\diff\lambda)$ had only continuous part, we would have that:
\begin{equation}\label{Fourier_Transform_Tom_chap_2}
\mathcal{W}_\sigma(X;\xi)= \mathcal{R}(\chi_\sigma)(X,\xi) = \langle R(X,\xi) , \chi_\sigma\rangle = \frac{1}{2\pi}\hspace{-.15cm}\int\limits_{\hspace{-.04cm}-\infty}^{\hspace{.45cm}\infty}\hspace{-0.15cm}\e^{-itX}\chi_\sigma\big(\hspace{-0.04cm}\exp(t\xi)\big)\diff t.
\end{equation}

Then, in such case, the following properties of the quantum tomograms $\mathcal{W}_\sigma$ are immediate consequences of Eq. \eqref{Fourier_Transform_Tom_chap_2}:

\begin{proposition}\label{tom_conds_prop}
Under the conditions stated previously, the quantum tomogram $\mathcal{W}_\sigma$ satisfies:
\begin{enumerate}

\item Normalization:  \hfil$\displaystyle{\hspace{-.25cm}\int\limits_{\hspace{-.04cm}-\infty}^{\hspace{.45cm}\infty}\hspace{-0.15cm}\mathcal{W}_\sigma(X;\xi)\, \diff X=1}$.

\item Homogeneity: \hfil$\displaystyle{\mathcal{W}_\sigma(sX;s\xi)=\frac{1}{s}\mathcal{W}_\sigma(X;\xi)},\quad s>0$.

\item Equivariance (\ref{eq:W_equivariance}): \hfil$\displaystyle{\mathcal{W}_\sigma(X; \mathrm{Ad\,}_{g^{-1}}\xi) = \mathcal{W}_{g^*\sigma}(X;\xi)}$.
\end{enumerate}
\end{proposition}


\subsection{A universal formula for the quantum Radon transform}\label{sec:QRT}
We will end this section by providing a universal formula that will allow a direct computation of the tomogram $\mathcal{W}_\sigma$ and that includes the standard quantum Radon transform (\ref{eq:quantum_radon_heisenberg}) as a particular instance.
This will justify why such expression could be called the \textit{Quantum Radon Transform} of a given state.

In order to provide such formula we will collect first a few facts on functional calculus of self-adjoint operators that will be needed in what follows (see for instance \cite{Da95}, \cite{Ak63} and references therein).

If $\mathbf{T}$ is a densely defined sefl-adjoint operator with spectral measure $E_\mathbf{T}(\diff \lambda)$,  given a measurable and integrable function $f$ with respect to the spectral measure $E_\mathbf{T}(\diff\lambda)$, we define the linear operator $f(\mathbf{T})$ as 
$$
f(\mathbf{T}) = \int f(\lambda)\, E_{\mathbf{T}}(\diff\lambda) \, .
$$   
The map $F_\mathbf{T} \colon L^1(\mathbb{R}, E_\mathbf{T}(\diff \lambda)) \to \mathbf{B}(\mathcal{H})$, given by $F(f) = f(\mathbf{T})$ satisfies: $F (f + g) = F(f) + F(g)$, $F(c f) = c F(f)$, $F(fg) = F(f) F(g)$ and is called a functional calculus \cite[chap. 4.4]{Ja68}.    

Given a real number $X$ we may also define the operator:
\begin{equation}\label{eq:fX-T}
f(X \mathbbm{1}  - \mathbf{T}) = \int f(X - \lambda) \, E_\mathbf{T} (\diff\lambda) \, .
\end{equation}   
Then, we may extend the notion of functional calculus introducing the symbol $\delta (X \mathbbm{1} - \mathbf{T})$ defined formally using Eq.  (\ref{eq:fX-T}):
\begin{equation}\label{eq:deltaX-T}
\delta (X\mathbbm{1} - \mathbf{T}) =  \hspace{-.02cm}\int\limits_{\hspace{-.20cm}-\infty}^{\hspace{.15cm}\infty}\hspace{-.0cm} \delta (X - \lambda) \, E_\mathbf{T} (\diff \lambda)  \, ,
\end{equation}
In order to make sense of the previous expression, we will consider it as an element in the dual of a class of bounded operators in the Hilbert space $\mathcal{H}$.   More specifically, as the previous expression will be evaluated in density operators, it makes sense to consider the family of trace class operators whose dual space is the space of all bounded operators with respect to the natural pairing provided by the trace:
\begin{equation}\label{eq:pairing}
\langle \mathbf{A}, \boldsymbol{\eta} \rangle = \mathrm{Tr\, }(\boldsymbol{\eta}^\dagger \mathbf{A}) \, , \qquad \mathrm{Tr\,}\boldsymbol{\eta} < \infty \, ,\quad \mathbf{A} \in \mathcal{B}(\mathcal{H}) \, .
\end{equation}
 In this sense we will define $\delta (X\mathbbm{1} - \mathbf{T})$ as an element in the dual of the space of trace class operators by means of:
$$
\langle \delta (X \mathbbm{1} - \mathbf{T}) , \boldsymbol{\sigma} \rangle =  \frac{1}{2\pi}\hspace{-.02cm}\int\limits_{\hspace{-.20cm}-\infty}^{\hspace{.15cm}\infty}\hspace{-.0cm} \diff k \int\limits_{\hspace{-.08cm}-\infty}^{\hspace{.15cm}\infty}\e^{ik(X- \lambda)} \mathrm{Tr\,} (\boldsymbol{\sigma} E_{\textbf{T}}(\diff\lambda))\, ,
$$
Notice that the r.h.s. in last expression can also be written as:
\begin{equation}\label{eq:deltaXT}
\langle \delta (X \mathbbm{1} - \mathbf{T}) , \boldsymbol{\sigma} \rangle =  \frac{1}{2\pi}\hspace{-.02cm}\int\limits_{\hspace{-.20cm}-\infty}^{\hspace{.15cm}\infty}\hspace{-.0cm} \diff k\, \,  e^{ikX} \mathrm{Tr\,} (\boldsymbol{\sigma} e^{-ik\mathbf{T}}) = 
 \frac{1}{2\pi} \mathrm{Tr\,} \left( \boldsymbol{\sigma}   \hspace{-.02cm}\int\limits_{\hspace{-.20cm}-\infty}^{\hspace{.15cm}\infty}\hspace{-.0cm} \diff k\, \,  e^{ik(X-k\mathbf{T})} \right) 
\, ,
\end{equation}
and, for instance, $\delta (X \mathbbm{1} - \mathbf{T})$ will be well defined provided that $\mathrm{Tr\,} (\boldsymbol{\sigma} e^{-ik\mathbf{T}}) \in L^2(\mathbb{R}, \diff k)$.

Then, using Eq. \eqref{eq:deltaXT}, with a slight abuse of notation, we can write:
$$
\delta(X \mathbbm{1} - \mathbf{T} ) = \frac{1}{2\pi}\hspace{-.02cm}\int\limits_{\hspace{-.20cm}-\infty}^{\hspace{.15cm}\infty}\hspace{-.0cm}  e^{ik(X \mathbbm{1} - \mathbf{T})} \diff k \, .
$$

We are ready to establish the main formula describing quantum tomograms $\mathcal{W}_\sigma$.  

\begin{theorem}\label{tom_cart}
Given a state $\rho$ in a unital $C^*$--algebra $\mathcal{A}$, the quantum tomogram $\mathcal{W}_{\sigma}(X;\xi)$ of any state $\boldsymbol{\sigma}$ in the folium of $\rho$ associated to the strongly continuous unitary representation $U$ of the Lie group $G$ on $\mathcal{A}$ is given by:
\begin{equation}\label{eq:TrsigmaX}
\mathcal{W}_{\sigma}(X;\xi)=\Tr\!\big(\boldsymbol{\sigma}\,\delta(X\mathbbm{1}-\langle\boldsymbol{\Theta},\xi\rangle)\big),\qquad\forall\xi\in\mathfrak{g},\, X\in\mathbb{R}.
\end{equation}
\end{theorem}

\begin{proof} 
First notice that $\mathrm{Tr\,} (\boldsymbol{\sigma} U_\rho(\exp (t\xi) ))$ is the smeared character $\chi_\sigma$ of the representation $U_\rho$ with respect to the state $\sigma$, then, because of Eq. (\ref{eq:deltaXT}), we get:
$$
\langle \delta (X \mathbbm{1} - \boldsymbol{\xi}) , \boldsymbol{\sigma} \rangle =  \frac{1}{2\pi}\hspace{-.02cm}\int\limits_{\hspace{-.20cm}-\infty}^{\hspace{.15cm}\infty}\hspace{-.0cm} \diff k\, \,  e^{ikX} \mathrm{Tr\,} (\boldsymbol{\sigma} e^{-ik\boldsymbol{\xi}}) = 
\frac{1}{2\pi}\hspace{-.02cm}\int\limits_{\hspace{-.20cm}-\infty}^{\hspace{.15cm}\infty}\hspace{-.0cm} \diff k\, \,  e^{ikX} \chi_\sigma (\exp (-k\xi)) \, . 
$$
and because of Eq. (\ref{Fourier_Transform_Tom_chap_2}), we get:
$$
\langle \delta (X \mathbbm{1} - \boldsymbol{\xi}) , \boldsymbol{\sigma} \rangle =  \mathcal{W}_\sigma (X, \xi) \, ,
$$
hence, substituting $\boldsymbol{\xi} = \langle\boldsymbol{\Theta},\xi\rangle$ and using the canonical pairing (\ref{eq:pairing}), we get the desired formula
$\mathcal{W}_{\sigma}(X;\xi)=\Tr\!\big(\boldsymbol{\sigma}\delta(X\mathbbm{1}-\langle\boldsymbol{\Theta},\xi\rangle)\big)$.
\end{proof}

Note that from the formula (\ref{eq:TrsigmaX}), we get easily that $\mathcal{W}_\sigma(X;\xi)$ is real:
\begin{multline}
\overline{\mathcal{W}_\sigma(X;\xi)}=\overline{\Tr\!\big(\boldsymbol{\sigma}\delta(X\mathbbm{1}-\langle\boldsymbol{\Theta},\xi\rangle)\big)}=\Tr\!\big(\boldsymbol{\sigma}\delta(X\mathbbm{1}^\dagger-\langle\boldsymbol{\Theta},\xi\rangle^\dagger)\big)\\
=\Tr\!\big(\boldsymbol{\sigma}\delta(X\mathbbm{1}-\langle\boldsymbol{\Theta},\xi\rangle)\big)=\mathcal{W}_\sigma(X;\xi) \, ,
\end{multline}
completing the list of conditions satisfied by quantum tomograms in Prop. \ref{tom_conds_prop}.


\subsection{Reconstruction of states from quantum tomograms on groups}\label{section_recosntruction_states_groups}

The lasts previous sections \ref{section_tomograms_group}, \ref{sec:QRT}, ended with the construction, using a group representation $U$ and a fiduciary state $\rho_0$, of tomograms $\mathcal{W}_\sigma$ for a family of normal states  $\boldsymbol{\sigma}$.  Because of \eqref{smeared_character_Fourier}, we can obtain the smeared character $\chi_\sigma$ of such states, as the inverse Fourier transform of the tomogram $\mathcal{W}_\sigma$:
\begin{equation}\label{smeared_character_Fourier_new}
\chi_\sigma\big(\hspace{-0.04cm}\exp(t\xi)\big) = \Tr\!\Big(\boldsymbol{\sigma}U_{\rho}\big(\hspace{-0.04cm}\exp(t\xi)\big)\Big)= \hspace{-.25cm}\int\limits_{\hspace{-.04cm}-\infty}^{\hspace{.45cm}\infty}\hspace{-0.15cm}\e^{itX}\mathcal{W}_\sigma(X;\xi)\diff X\, .
\end{equation}
What is needed now is to recover the state $\boldsymbol{\sigma}$ from the smeared character $\chi_\sigma$, in other words, we would like to know under what conditions the representation $U$ determines a tomographic pair. 

The notion of frames \cite{Da92}, \cite{Ka94}, a widely studied subject with many applications\footnote{Even if we will focus here on their use in sampling theory, very close in spirit to the problem we are facing here  (see, for instance, \cite{Ga02} and references therein for classical sampling theory, and \cite{Ga08}, for the use of frames in such context).}, will provide the key idea to achieve this.   In the following paragraphs, we will succinctly review the main notions of the theory that will be needed in this context.

Let $\mathcal{H}$ be a complex separable Hilbert space and $\MM$ a measurable space with a measure $\mu$.  A family of vectors $\mathcal{F} = \{ |\psi_x\rangle  \mid x \in \MM \}$ is called a frame based on $\MM$ if it satisfies:
\begin{enumerate}
\item For every $|\psi \rangle \in \mathcal{H}$, the function $\mathrm{ev}_\psi \colon \MM \to \mathbb{C}$ given by $\mathrm{ev}_\psi (x) = \langle \psi \mid \psi_x \rangle$ is $\mu$-measurable and belongs to $L^2(\MM,\mu)$ (we may also say that the frame $\mathcal{F}$ is square integrable).

\item (Stability condition) There are real numbers $0 < A \leq B$ such that:
$$
A ||\psi ||^2 \leq || \mathrm{ev}_\psi ||^2_{L^2(\MM)} = \int_{\MM} |\langle \psi, \psi_x \rangle |^2 \diff \mu (x) \leq B ||\psi ||^2 \, ,
$$
for all $|\psi \rangle \in \mathcal{H}$.  The numbers $A$ and $B$ in the previous expression are called (lower ande upper) frame bounds for the frame $\mathcal{F}$, respectively.  If $A =B$, the frame is called tight.
\end{enumerate}

A frame $\mathcal{F} = \{ |\psi_x\rangle  \mid x \in \MM \}$ defines a bounded linear operator, called the frame operator, $\mathfrak{F} \colon \mathcal{H} \to L^2(\MM,\mu)$, given by $\mathfrak{F}|\psi \rangle = \langle \psi_{(\cdot)} | \psi \rangle$, in other words, the function $\mathfrak{F}|\psi \rangle \colon \mathcal{M} \to \mathbb{C}$, denoted otherwise as $F_\psi$, is defined as $(\mathfrak{F} |\psi \rangle ) (x) = F_\psi(x) = \langle \psi_x | \psi \rangle$, $x \in \MM$. Moreover, the frame operator $\mathfrak{F}$ is injective and admits a bounded left inverse.   

The adjoint frame operator $\mathfrak{F}^* \colon L^2(\MM,\mu) \to \mathcal{H}$ is defined as: 
$$
\langle \mathfrak{F}^* f | \psi \rangle_\mathcal{H} = \langle f , \mathfrak{F}|\psi \rangle_{L^2(\mathcal{M})}\, ,\qquad   \forall  f \in L^2(\MM)\,,|\psi \rangle \in \mathcal{H} \, .
$$
Then, 
\begin{eqnarray*}
\langle \mathfrak{F}^* f | \psi \rangle_\mathcal{H} &=& \int_{\MM} \overline{f(x)} F_\psi(x)\, \diff \mu(x) \\ &=& \int_{\MM} \overline{f(x)} \langle \psi_x | \psi \rangle \, \diff\mu(x)  =  \langle \Big(\int_{\MM} f(x) \psi_x \, \diff\mu(x)  \Big)| \psi \rangle  \, ,
\end{eqnarray*}
and
$$
\mathfrak{F}^* f = \int_{\MM} f(x) | \psi_x \rangle \, \diff\mu(x) \, , \qquad \forall f \in L^2(\MM,\mu) \, ,
$$
in the weak sense (as all other operator-valued or vector-valued integrals we have encountered so far).

Using the frame operator and its adjoint we can define the metric operator $S = \mathfrak{F}^* \mathfrak{F} \colon \mathcal{H} \to \mathcal{H}$, which is bounded and definite positive, with bounded definite positive inverse $S^{-1}$, and satisfying $0 < A I \leq S \leq B I$.  

The metric operator $S$ allows to define the dual frame $\mathcal{F}^*$ of the frame $\mathcal{F} = \{ |\psi_x\rangle  \mid x \in \MM \}$, namely, the family of vectors:
$$
\mathcal{F}^* = \{ |\psi^x \rangle = S^{-1} |\psi_x\rangle  \mid  x \in \MM \} \, ,
$$
which is actually a frame too.   We will also say that the frames $\mathcal{F}$ and $\mathcal{F}^*$ are dual to each other.

Using a couple of dual frames  $\mathcal{F}$ and $\mathcal{F}^*$, we obtain some remarkable formulas:
$$
I = \int_{\MM} |\psi^x \rangle \langle \psi_x |\, \diff \mu (x) = \int_{\MM} |\psi_x \rangle \langle \psi^x | \, \diff \mu (x) \, ,
$$
and we get the following reconstruction formula (compare with Eq. (\ref{reconstruction_state_D})):
$$
|\psi \rangle  = \int_\mathcal{M} F_\psi (x) |\psi^x \rangle \, \diff \mu(x)  \, , \qquad \forall |\psi \rangle \in \mathcal{H} \, .
$$
together with a similar formula for the dual frame transform.

The main observation regarding the use of frames in the context of quantum tomography is that a relevant class of unitary representations of groups, the so called square integrable, define frames.  A square integrable representation $\pi$ of a locally compat group $G$ on a Hilbert space $\mathcal{H}$ is an irreducible unitary representation such that there exist vectors $|\psi \rangle, |\phi \rangle \in \mathcal{H}$ for which:
$$
\int_G | \langle \phi ,\pi (g) \psi \rangle|^2 \diff \mu_G (g)  < \infty \, .
$$
with $\mu_G$ the left-invariant Haar measure on $G$.  Sometimes the function $c_{\phi,\psi}(g) = \langle \phi ,\pi (g) \psi \rangle$ is called a coefficient of the representation and, with this terminology, the representation is square integrable if there is a square integrable coefficient function on $G$.

If a representation $\pi$ is square integrable, the set of vectors for which the function $ \langle \phi |\pi (g) |\psi \rangle$ is in $L^2(G, \mu_G)$ define a dense linear subspace of $\mathcal{H}$.  Moreover if the group $G$ is unimodular, i.e., the left-invariant Haar measure is right-invariant, then the function $\langle \phi |\pi (g) |\psi \rangle$ is square integrable for all $|\psi\rangle$, $|\phi\rangle$.  More specifically we have the following Theorem \cite[Thm. 14.3.3]{Di77} (see also \cite{Du76},  \cite{An08} for the non-unimodular case) suited for the purposes of this paper. 

\begin{theorem}\label{thm:square} Let $\pi \colon G \to \mathcal{U}(\mathcal{H})$ be  a square integrable irreducible unitary representation of the unimodular Lie group $G$, then the coefficient function $c_{\phi,\psi} = \langle \phi, \pi (.) \psi \rangle$ is  in $L^2(G, \mu_G)$ for all  $|\psi\rangle$, $|\phi\rangle \in \mathcal{H}$, and the orthogonality relations:
\begin{equation}\label{eq:orthogonality}
\int_G \langle \phi_1 ,\pi (g) \psi_1 \rangle \langle \pi (g) \phi_2 , \psi_2\rangle \diff \mu_G(g) =d^{-1} \langle \phi_1 , \phi_2 \rangle \langle \psi_1 , \psi_2\rangle \, ,
\end{equation}
hold for all $|\psi_a\rangle$, $|\phi_a\rangle$, $a = 1,2$, where $0< d < +\infty$ is a positive real number called the formal degree of the representation.
\end{theorem}

An immediate consequence of Thm. \ref{thm:square} is that, if $G$ is a unimodular Lie group, $U \colon G \to \mathcal{A}$ a unitary representation of $G$, and $\rho_0$ a state such that the unitary representation $U_0  = \pi_{\rho_0}\circ U$ on the GNS Hilbert space $\mathcal{H}_{\rho_0}$ is square integrable, then the family of vectors $U_0 (g) |\psi \rangle$ is a frame.

\begin{proposition}
Let $U\colon G \to \mathcal{A}$ be a strongly continuous unitary representation of the unimodular Lie group $G$ on the $C^*$-algebra $\mathcal{A}$ and $\rho_0$ a pure fiducial state on $\mathcal{A}$.  Consider the associated unitary representation $U_0 \colon G \to \mathcal{U}(\mathcal{H}_{\rho_0})$ of the group $G$ on the GNS Hilbert space $\mathcal{H}_{\rho_0}$ defined by the state $\rho_0$ and assume that $U_0$ is square integrable with formal degree $d$.  Then, given a vector $|\psi_0 \rangle$, the set of vectors $\mathcal{F}_{U,\rho_0}  = \{ |\psi_g \rangle = U_0(g) |\psi_0 \rangle \mid g \in G\}$ defines a tight frame, whose dual frame is given by $\mathcal{F}^*_{U,\rho_0} =  \{ |\psi^g \rangle = \frac{d}{||\psi_0||^2}\, U_0(g) |\psi_0 \rangle \mid g \in G\}$.
\end{proposition}

\begin{proof}
If the fiducial state $\rho_0$ is pure, the GNS representation $\pi_{\rho_0}$ is irreducible \cite{Ha96} and, consequently, the unitary representation $U_0 = \pi_{\rho_0} \circ U$ is irreducible.  Then, because of Thm. \ref{thm:square}, the orthogonality relations (\ref{eq:orthogonality}) imply:
$$
\int_G |\langle \psi , \psi_g \rangle |^2 \, \diff \mu_G(g) = \int_G |\langle \psi , U_0(g) \psi_0 \rangle |^2  \, \diff \mu_G(g) = d^{-1} || \psi_0 ||^2 || \psi ||^2 \, ,
$$
for all $|\psi \rangle \in \mathcal{H}$, and the set $\{ |\psi_g \rangle \}$ is a tight frame with frame constant $d^{-1} || \psi_0 ||^2$.

A simple computation shows that in such case the frame map $\mathfrak{F} \colon \mathcal{H} \to L^2(G,\mu_G)$ is an isometry (note the the range of $\mathfrak{F}$ is a closed subspace invariant under $G$, hence because the representation $U_0$ is irreducible it must be $\mathcal{H}$), and the map $S$ is a multiple of the identity with scaling factor the frame constant.
\end{proof}

In what follows, in order to obtain simpler formulas, it will be assumed that the vector $|\psi_0\rangle$ is unitary, that is, $||\psi_0 || = 1$.   

We will end this summary by recalling the trace theorem for frames that states that the range of the frame operator $\mathfrak{F}(\mathcal{H}) \subset L^2(\MM,\mu)$ is a reproducing kernel Hilbert space with kernel given by the function $\kappa (x,y) = \langle \psi_x, \psi^y \rangle$, $x,y \in \MM$.  This implies that (compare with Eq. (\ref{eq:reproducing})):
$$
\Phi (y) = \int_{\MM} \kappa (y,x) \Phi (x)\, \diff \mu (x) \, , \qquad \forall \Phi \in \mathfrak{F}(\mathcal{H}) \subset L^2(\MM, \mu)\, ,
$$
that in our particular setting becomes:
\begin{equation}\label{eq:reproducingG}
\Phi (g) = \int_G \kappa (g,g') \Phi (g')\, \diff \mu_G (g') \, , \qquad \forall \Phi \in \mathfrak{F}(\mathcal{H}) \subset L^2(G,\mu_G) \, .
\end{equation}
with the function,
 $\kappa (g,g') = d\, \langle U_0(g)\psi_0, U_0(g')\psi_0\rangle = d\, \langle \psi_0, U_0(g^{-1}g')\psi_0\rangle$, $g,g' \in G$, called the reproducing kernel, which is $L^2 (G,\mu_G)$ on each argument.   Note that an immediate consequence of Eq. (\ref{eq:reproducingG}) is that:
 \begin{equation}\label{eq:kappakappa}
 \kappa (g,g' ) = \int_G \kappa (g,g'') \kappa (g'', g')\, \diff \mu_G(g'') \, .
 \end{equation}
 From the previous expressions we will obtain the following formula for the trace of a trace class operator $\mathbf{A}$:
 \begin{equation}\label{eq:trace}
 \mathrm{Tr\,} \mathbf{A} = \int_G \kappa (\mathbf{A},g,g)\, \diff \mu_G(g) \, , 
 \end{equation}
 where $ \kappa (\mathbf{A},g,g') =  \langle U_0(g)\psi_0, \mathbf{A} U_0(g')\psi_0\rangle$.  In particular if we choose $\mathbf{A} = |\psi_g \rangle \langle \psi_{g'} |$, $g,g' \in G$, we get:
 \begin{eqnarray}\label{eq:Trpsigpsig}
 \mathrm{Tr\,} (|\psi_g \rangle \langle \psi_{g'}| ) &=& \int_G 
 \langle \psi_0 | U(g'')^\dagger |\psi_g \rangle \langle \psi_{g'} | U(g'') |\psi_0 \rangle\, \diff \mu_G (g'')  \nonumber \\ &=& \int_G \kappa (g'',g) \kappa(g', g'') \diff \mu_G(g'') = \kappa (g',g) \, ,
\end{eqnarray}
where we have used (\ref{eq:kappakappa}) to get the final formula.

We will apply now the previous results to the situation described in Sects. \ref{sec:dynamics}, \ref{section_tomograms_group}.   That is, because the states $\boldsymbol{\sigma}$ we want to describe tomographically are the states in the folium of $\rho_0$, then we will consider the $C^*$-algebra $\mathcal{A}_0 = \mathcal{B}(\mathcal{H}_{\rho_0})$ of all bounded operators in $\mathcal{H}_{\rho_0}$  whose states are represented by density operators $\boldsymbol{\sigma}$.  

The unimodular Lie group $G$ is represented in the Hilbert space $\mathcal{H}_{\rho_0}$ by a irreducible square integrable unitary representation $U_0$ of formal degree $d$.     Then we may consider the map $U_0 \colon G \to \mathcal{A}_0$ as a tomographic set and $D_0 \colon G\to \mathcal{A}_0'$, defined as $D_0(g) = d\,  U_0(g)^\dagger$, $g \in G$, as the dual tomographic map.    That these notions correspond to the general scheme depicted in Sect. \ref{sec:sampling} is the content of the following resuts.

\begin{theorem}
The pair $U_0$, $D_0$ is a tomographic pair.
\end{theorem}

\begin{proof}  
The assertion follows from the direct proof of the reconstruction formula given in Thm. \ref{thm:reconstruction_density} below.  

We observe that the set $\mathcal{U}_0  = \{  U_0(g) \in \mathcal{A}_0 \mid g \in G\}$ separates states in the folium of $\rho_0$.  Note that if this were not the case, there will be two states $\boldsymbol{\sigma}$, $\boldsymbol{\sigma}'$, such that $\mathrm{Tr\, } (\boldsymbol{\sigma} U(g)) = \mathrm{Tr\, } (\boldsymbol{\sigma}' U(g)) $ for all $g\in G$, i.e., $\mathrm{Tr\, } ((\boldsymbol{\sigma}- \boldsymbol{\sigma}') U(g)) = 0$ for all $g \in G$.  But the representation $U_0$ is irreducible, then  $(\boldsymbol{\sigma}- \boldsymbol{\sigma}')  = 0$.

Moreover, because the representation $U_0$ is square integrable, the function $\chi_\sigma (g) = \mathrm{Tr\,}(\boldsymbol{\sigma} U(g))$ is square integrable and the map $\mathcal{F}_{U_0}(\boldsymbol{\sigma}) (g) = \mathrm{Tr\,}(\boldsymbol{\sigma} U_0(g))$ maps states in square integrable functions.  Finally, because of Eq. (\ref{eq:reproducingG}), the reproducing property (\ref{eq:reproducing}) of the kernel is verified.  
\end{proof}

Then we will proof the reconstruction formula, the particular instance of the general reconstruction equation, Sect. \ref{sec:reconstruction}, Eq. (\ref{reconstruction_state_D}),  for density operators:

\begin{theorem}\label{thm:reconstruction_density}
Under the conditions stated previously, we have:
\begin{equation}\label{eq:reconstruction_operator}
\boldsymbol{\sigma} = \int_G \mathrm{Tr\,}(\boldsymbol{\sigma} U_0(g)) D_0(g)\, \diff \mu_G(g) = d \int_G \mathrm{Tr\,}(\boldsymbol{\sigma} U_0(g)) U_0(g)^\dagger \diff \mu_G(g) \, ,
\end{equation}
for all density operators $\boldsymbol{\sigma}$ in $\mathcal{H}_{\rho_0}$.
\end{theorem}

In order to proof this theorem, we will need the following extension of the trace formula (\ref{eq:trace}):

\begin{proposition}
Let $\mathbf{A}$ be a trace class operaton on $\mathcal{H}_{\rho_0}$, then, under the conditions stated previously, we have:
\begin{equation}\label{eq:trace_operator}
\mathrm{Tr\,}(\mathbf{A}) \, \mathbbm{1} = d\, \int_G U_0(g) \mathbf{A} U_0(g)^\dagger \diff \mu_G(g) \, .
\end{equation}
\end{proposition}

\begin{proof}
Notice that, because the group $G$ is unimodular and the representation $U_0$ is square integrable, the operator 
$$
C = \int_G U_0(g) \mathbf{A} U_0(g)^\dagger \diff \mu_G(g) \, ,
$$ 
is well-defined and intertwines the representation $U_0$, i.e., $C U_0(g) = U_0(g) C$, for all $g \in G$, then because $U_0$ is irreducible $C$ is a scalar multiple of the identity, $c \mathbbm{1}$. Then $\langle \psi_0 , c \mathbbm{1} \psi_0 \rangle = c d$ with $|\psi_0\rangle$ unitary. Hence, because of (\ref{eq:trace}), we get:
$$
c d = \langle \psi_0, \int_G U_0(g) \mathbf{A} U_0(g)^\dagger \diff \mu_G(g) \psi_0 \rangle =  \int_G \kappa( \mathbf{A} ,g,g) \diff \mu_G(g) = \mathrm{Tr\, }(\mathbf{A}) \, ,
$$
and the conclussion follows.
\end{proof}

Noe we can prove Thm. \ref{thm:reconstruction_density} (the proof is inspired in the proof in \cite[Appendix]{Da03b}).

\begin{proof}(Theorem \ref{thm:reconstruction_density})   We will denote by $\boldsymbol{\sigma}'$ the r.h.s. of Eq. (\ref{eq:reconstruction_operator}), that is:
\begin{equation}\label{eq:sigma'}
\boldsymbol{\sigma}' = d\, \int_G \mathrm{Tr\,}(\boldsymbol{\sigma} U_0(g)) U_0(g)^\dagger \diff \mu_G(g) \, ,
\end{equation}
If we multiply the previous Eq. (\ref{eq:sigma'}) on the left  by a trace class operator $\mathbf{A}$ and we apply the  trace formula Eq. (\ref{eq:trace_operator}) to the r.h.s. we get:
\begin{eqnarray}\label{eq:main_eq_proof}
&& \int_G \mathrm{Tr\,}(\boldsymbol{\sigma} U_0(g)) \mathbf{A} U_0(g)^\dagger \diff \mu_G(g) =  \nonumber\\
&& = \int_G \int_G U_0(g') \boldsymbol{\sigma}U_0(g)  U(g')^\dagger \mathbf{A} U_0(g)^\dagger \diff \mu_G(g) \diff \mu_G(g') \, .
\end{eqnarray}
Then, using Fubini's theorem and the trace formula again on the r.h.s. of (\ref{eq:main_eq_proof}), we get:
\begin{equation}\label{eq:intermediate}
\int_G \mathrm{Tr\,}(\boldsymbol{\sigma} U_0(g)) \mathbf{A} U_0(g)^\dagger \diff \mu_G(g) = \int_G \mathrm{Tr\,}(U_0(g')^\dagger \mathbf{A}) U_0(g') \boldsymbol{\sigma}  \diff \mu_G(g') \,.
\end{equation}
Then, using the previous equation (\ref{eq:intermediate}), we get:
\begin{equation}\label{eq:last_step}
\mathrm{Tr\,}(\mathbf{A} \boldsymbol{\sigma}') =  \int_G \mathrm{Tr\,}(\mathbf{A} U_0(g')^\dagger )  \mathrm{Tr\,}(U_0(g') \boldsymbol{\sigma} ) \diff \mu_G(g') \,.
\end{equation}
where we have used the circularity of the trace $\mathrm{Tr\,}(\mathbf{A} U_0(g')^\dagger) = \mathrm{Tr\,}(U_0(g')^\dagger\mathbf{A} )$ (note that $\mathbf{A} U_0(g)$ is a trace-class operator).  Expanding the traces inside the r.h.s. of Eq. (\ref{eq:last_step}), integrating with respect to $g'$, and using the identities (\ref{eq:Trpsigpsig}) and (\ref{eq:trace_operator}), we obtain:
\begin{eqnarray*}
&& \int_G \mathrm{Tr\,}(\mathbf{A} U_0(g')^\dagger )  \mathrm{Tr\,}(U_0(g') \boldsymbol{\sigma} ) \diff \mu_G(g') = 
\\ && =  \int_G \langle \psi_g | \mathbf{A} U_0(g')^\dagger | \psi_g \rangle  \langle \psi_{g''} | U_0(g') \boldsymbol{\sigma} |\psi_{g''}\rangle \diff \mu_G(g') \diff \mu_G(g) \diff \mu_G (g'') = \\
&& = \int_G \langle \psi_g | \mathbf{A} \left( \int_G U_0(g')^\dagger | \psi_g \rangle  \langle \psi_{g''} | U_0(g') \diff \mu_G(g') \right) \boldsymbol{\sigma} |\psi_{g''}\rangle \diff \mu_G(g) \diff \mu_G (g'')  = \\
&& = \int_G \langle \psi_g | \mathbf{A}\, \kappa (g'',g) \boldsymbol{\sigma} |\psi_{g''}\rangle \diff \mu_G(g) \diff \mu_G (g'') =
\\ && =  \frac{1}{d} \int_G \kappa (g'',g) \kappa (\mathbf{A} \boldsymbol{\sigma}, g, g'' ) \diff \mu_G(g) \diff \mu_G (g'')  \\ && = \frac{1}{d} 
 \int_G \kappa (\mathbf{A} \boldsymbol{\sigma}, g, g ) \diff \mu_G(g) = \frac{1}{d}  \mathrm{Tr\,} ( \mathbf{A} \boldsymbol{\sigma} )\, .
\end{eqnarray*}
where, in the fourth step, we have used the identities (\ref{eq:trace_operator}) and (\ref{eq:Trpsigpsig}) to compute $ \int_G U(g')^\dagger | \psi_g \rangle  \langle \psi_{g''} | U_0(g') \diff \mu_G(g') $, and  the identities (\ref{eq:kappakappa}) and (\ref{eq:trace}) in the  last two steps. Then we conclude that:
$$
\mathrm{Tr\,}(\mathbf{A} \boldsymbol{\sigma}')  = \mathrm{Tr\,} ( \mathbf{A} \boldsymbol{\sigma} )\, ,
$$
for all trace class operators $\mathbf{A}$, but Hilbert-Schmidt operators are contained in the ideal of trace-class operators, and density operators are Hilbert-Schmidt operators, then the last identity can be written as:
$$
\langle \mathbf{A}, \boldsymbol{\sigma}' \rangle  = \langle \mathbf{A} ,\boldsymbol{\sigma} \rangle\, ,
$$
for all $\mathbf{A}$ Hilbert-Schmidt operators, and we conclude that $\boldsymbol{\sigma}'  = \boldsymbol{\sigma}$.
\end{proof}


\section{Examples and applications}\label{sec:examples}

\subsection{Compact and finite groups}

\subsubsection{Compact groups}

 Compact groups, being unimodular (the left and right Haar measures coincide), are particularly well suited to provide tomographic descriptions of quantum states by applying the theory developed in Sec. \ref{section_particular group}.
Moreover any irreducible representation is finite-dimensional and square integrable. 

Then, if $U_0$ is an irreducible representation of the compact Lie group $G$ on the finite-dimensional Hilbert space $\mathcal{H}_0$ of dimension $n$, the orthogonality relations \eqref{eq:orthogonality} become:
\begin{equation}\label{eq:orthogonality_compact}
\int_G \langle \phi_1 ,U_0 (g) \psi_1 \rangle \langle U_0 (g) \phi_2 , \psi_2\rangle \diff \mu_G(g) = \frac{1}{n} \langle \phi_1 , \phi_2 \rangle \langle \psi_1 , \psi_2\rangle \, ,
\end{equation}
with $\mu_G$ the normalized Haar measure on the group, that is, $\int_G\diff\mu_G(g)=1$.  
We will use the representation $U_0$ as the tomographic set on the $C^*$-algebra $\mathcal{B}(\mathcal{H}_0)$ (which is finite-dimensional).   States are density operators $\boldsymbol{\sigma}$, that is, normalized, Hermitean non-negative operators.
The sampling function, or smeared character of the representation $U_0$, associated with the state $\boldsymbol{\sigma}$ will be $\chi_\sigma (g) = \mathrm{Tr\,}(\boldsymbol{\sigma}U_0(g))$ and it defines a tomographic set.    
If we define $D(g) =  n U(g)^\dagger$, we will have that the kernel function $\kappa$ will be given by:
$$
\kappa (g,g') = n \mathrm{Tr\,}(U_0(g^{-1}g') ) \, ,
$$
amd the reconstruction formula 
for density operators $\boldsymbol{\sigma}$ on $\mathcal{H}_0$ becomes:
\begin{equation}\label{rec_states_compact_schur}
\boldsymbol{\sigma} = n\hspace{-0.15cm}\int\limits_{\hspace{.3cm}G}\hspace{-0.1cm}\chi_\sigma(g)U(g)^\dagger\diff\mu(g) \, .
\end{equation}

Tomograms were defined, recall formula \eqref{eq:TrsigmaX}, as $\mathcal{W}_\sigma (X, \xi) = \mathrm{Tr\,}(\boldsymbol{\sigma}(X \mathbbm{1} - \boldsymbol{\xi}))$.   The spectral measure associated with the operator $\boldsymbol{\xi}$ is singular and concentrated on the (finite) spectrum $\sigma (\boldsymbol{\xi})$ of the operator $\boldsymbol{\xi}$, i.e., $E_{\boldsymbol{\xi}} = \sum_{\lambda \in \sigma ( \boldsymbol{\xi})}  P_\lambda\, \delta (\lambda' - \lambda)$,  with $P_\lambda$ the orthogonal projector onto the eigenspace of the eigenvalue $\lambda$. Then, it is easy to check that the operator $\delta (X \mathbbm{1} - \boldsymbol{\xi})$ becomes:
\begin{equation}\label{eq:delta_compact}
\delta (X \mathbbm{1} - \boldsymbol{\xi}) = \sum_{\lambda \in \sigma ( \boldsymbol{\xi})} \delta_{X,\lambda}  P_\lambda \, ,
\end{equation}
with $\delta_{X,\lambda} = 1$, if $X = \lambda$, and zero otherwise.  Let us summarize these results as the following corollary:

\begin{corollary}\label{reconstruction_thm}
Let $G$ be a compact connected Lie group and $(\mathcal{A},U)$ a unitary representation of $G$ on a $C^*$--algebra $\mathcal{A}$.
Let $\rho_0$ be a fiducial state such that the unitary representation $U_0(g)=\pi_{\rho_0}\big(U(g)\big)$ of $G$ on the GNS Hilbert space $\mathcal{H}_0$ defined by $\rho_0$ is irreducbie.  Let $n$ be the dimension of $\mathcal{H}_0$.
Then, given a density operator $\boldsymbol{\sigma}$ on $\mathcal{H}_0$, its tomograms are given by:
\begin{equation}\label{eq:tom_compact}
\mathcal{W}_{\boldsymbol{\sigma}} (X, \xi) = \sum_{\lambda \in \sigma ( \boldsymbol{\xi})} \delta_{X,\lambda}  \mathrm{Tr\,}(\boldsymbol{\sigma} P_\lambda) , \, \qquad  \forall \xi \in\mathfrak{g} \, ,
\end{equation}
and the state can be reconstructed from its tomograms by means of
\begin{equation}\label{eq:recon_tom_com}
\boldsymbol{\sigma} = n \int \e^{itX}\mathcal{W}_\sigma(X;\xi) U_0(\exp (-t\xi)) \Delta (t, \xi) \, \diff X \, \diff t \, \diff \Omega (\xi) \, ,
\end{equation}
with $\diff \Omega (\xi)$ the Lebesgue measure on the Lie algebra $\mathfrak{g}$, and $ \Delta (t, \xi)$ the Jacobian of the restriction of the exponential map to the unit sphere on $\mathfrak{g}$. 
 \end{corollary}


\subsubsection{Spin Tomography: the case of the $SU(2)$ group}\label{section_Spin_Tomography}

The compact group $SU(2)$ is the group that underlies the description of the states of a particle with spin $1/2$\,, (see for instance \cite[Chap. 5]{Ga90}, \cite[Chap. 11]{Es04}).  It also provides the background for the Jordan-Schwinger map that makes it so relevant in quantum optics and in quantum information theory. The states of a particle with spin $1/2$ may be represented using the so called \textit{Bloch's sphere}, Fig. \ref{B_sphere}.
\begin{figure}[h]
\centering
\includegraphics[width=5cm]{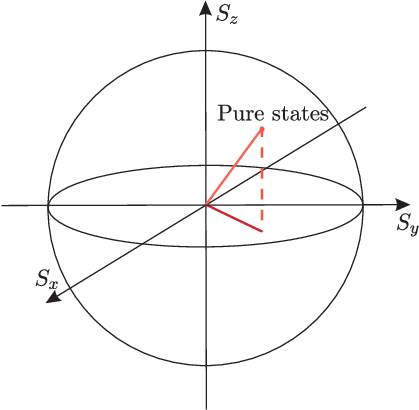}
\caption{Bloch's sphere representing states of a particle with spin $1/2$.}
\label{B_sphere}
\end{figure}
e.g., density states can be parametrized as:
\begin{equation}
\boldsymbol{\rho}=\frac{1}{2}\begin{pmatrix}
1+r\cos\theta & \sin\theta\e^{-i\phi}\\
\sin\theta\e^{i\phi} & 1-r\cos\theta
\end{pmatrix},
\end{equation}
with $0\leq r\leq 1$, $0\leq\theta\leq\pi$ and $0\leq\phi\leq 2\pi$. The states in the surface, i.e., the states with $r=1$, are the pure states of the system.

The $SU(2)$ group is a compact Lie group, hence given an irreducible representation of the group,  its tomograms are given by (\ref{eq:tom_compact}) and we can reconstruct a given state $\boldsymbol{\rho}$ by using the reconstruction equation \eqref{eq:recon_tom_com}.  The irreducible representations of the group $SU(2)$ are finite-dimensional of dimension $2j +1$ where $j$ is a half-integer.    The irreducible representation of dimension 2, defines the group itself as the set of $2\times 2$ unitary matrices $U$ of determinant 1 acting on $\mathbb{C}^2$ by standard matrix multiplication.  The elements of the group $U$ can be written in terms of the exponential of the elements of its Lie algebra $\mathfrak{su}(2)$ as:
\begin{equation}
U(s_x,s_y,s_z)=\e^{i(s_x\textbf{S}_x+s_y\textbf{S}_y+s_z\textbf{S}_z)},
\end{equation}
where the generators corresponding to the spin along the axis $x,y,z$ are:
\begin{equation}\label{spin_operators_tom}
\textbf{S}_a = \frac{\hbar}{2}\sigma_a\, ,  \qquad a = x,y,z\, .
\end{equation}
and $\sigma_a$, $ a = x,y,z$, are the Pauli matrices:
\begin{equation}\label{pauli_matrices}
\sigma_x=\begin{pmatrix}
0 & 1\\
1 & 0
\end{pmatrix},\quad\sigma_y=\begin{pmatrix}
0 & -i\\
i & 0
\end{pmatrix},\quad\sigma_z=\begin{pmatrix}
1 & 0\\
0 & -1
\end{pmatrix}.
\end{equation}

The tomograms of a state $\boldsymbol{\rho}$ will be:
\begin{equation}
\mathcal{W}_\rho(X,s_x,s_y,s_z)=\Tr\!\Big(\boldsymbol{\rho}\, \delta\big(X\mathbbm{1}-\frac{s_x}{2}\sigma_x-\frac{s_y}{2}\sigma_y-\frac{s_z}{2}\sigma_z\big)\Big).
\end{equation}
Then using Eq. (\ref{eq:delta_compact}) and taking into account that the eigenvalues of the operator $\mathbf{S} = \boldsymbol{s} \cdot \boldsymbol{\sigma} = s_x \sigma_x +  s_y \sigma_y + s_z \sigma_z$ are $\lambda=\pm|\boldsymbol{s}|/2$, with eigenvectors:\begin{equation}
|v_{+}\rangle=\frac{1}{\sqrt{2|\boldsymbol{s}|(|\boldsymbol{s}|-s_z)}}
\begin{pmatrix}
-s_x+is_y\\
s_z-|\boldsymbol{s}|
\end{pmatrix}\quad\text{and}\quad |v_{-}\rangle=\frac{1}{\sqrt{2|\boldsymbol{s}|(|\boldsymbol{s}|-s_z)}}
\begin{pmatrix}
s_z-|\boldsymbol{s}|\\
s_x+is_y
\end{pmatrix} \, ,
\end{equation}
the tomogram $\mathcal{W}_\rho(X,s_x,s_y,s_z)$ can be written as:
\begin{equation}
\mathcal{W}_\rho(X,s_x,s_y,s_z) = \delta(X -  \frac{1}{2}|\boldsymbol{s}|                                                                                                                                                                                                                                                                                                                                                                                                                                                                                                                                                                                                                                                                                                                                                                                                                                                                                                    \langle v_{+}|\boldsymbol{\rho}|v_{+}\rangle +
\delta(X + \frac{1}{2}|\boldsymbol{s}|                                                                                                                                                                                                                                                                                                                                                                                                                                                                                                                                                                                                                                                                                                                                                                                                                                                                                                    \langle v_{-}|\boldsymbol{\rho}|v_{-}\rangle \, ,
\end{equation}
and therefore, we finally get:
\begin{multline}
\mathcal{W}_\rho(X,s_x,s_y,s_z)=
\frac{1}{2}\left(\Big(1+\frac{s_z}{|\boldsymbol{s}|}r\cos\theta\Big)+\sin\theta\Big(\frac{s_x}{|\boldsymbol{s}|}\cos\phi+\frac{s_y}{|\boldsymbol{s}|}\sin\phi\Big)\right) \delta (X - \frac{1}{2}|\boldsymbol{s}|) \nonumber\\
\hspace{-0.6cm}+\frac{1}{2}\left(\Big(1-\frac{s_z}{|\boldsymbol{s}|}r\cos\theta\Big)-\sin\theta\Big(\frac{s_x}{|\boldsymbol{s}|}\cos\phi+\frac{s_y}{|\boldsymbol{s}|}\sin\phi\Big)\right)\delta (X + \frac{1}{2}|\boldsymbol{s}|)\, . 
\end{multline}


\subsubsection{Finite groups}
Finite groups are particular instances of compact groups where the previous analysis becomes particularly simple albeit interesting on itself.

Let $(U, \mathcal{H})$ be an irreducible unitary representations of the finite group $G$ of order $|G|$ on the finite-dimensional Hilbert space $\mathcal{H}$, $n=\textrm{dim}(\mathcal{H})$.  Let $e_i$, $i=1,\ldots,n$, be a given orthonormal basis on such space. We will denote by $U_{ij}(g)$ the elements of the unitary matrix associated to $U(g)$, $g\in G$, in the previous basis, i.e.,
\begin{equation}
U_{ij}(g)=\langle e_i,U(g)e_j\rangle \, .
\end{equation}

Schur's orthogonality relations is the form that the orthogonality relations for square integrable representations (\ref{eq:orthogonality}) take in the case of finite groups, and they assert that given two unitary irreducible representations $U^{(a)}$ and $U^{(b)}$ of dimensions $n_a$ and $n_b$, respectively, then:
\begin{equation}\label{schur_ortho}
\frac{n_a}{|G|}\sum_{g\in G} U_{ij}^{(a)}(g)^\dagger U_{rs}^{(b)}(g)=\delta_{ab}\delta_{ir}\delta_{js} \, .
\end{equation}
Therefore, if we choose as the dual tomographic map the Hermitean conjugate of the $n$-dimensional irreducible representattion $U(g)$, $D(g)= n \, U(g)^\dagger$, we will get:
\begin{equation}\label{biortho_finite_groups}
\frac{1}{|G|}\sum_{g\in G} D_{ij}(g) U_{jk}(g)=\delta_{ik} \, ,
\end{equation}
from which the biorthogonality condition for the pair $U, D$, is easily derived. 
Hence, if $\boldsymbol{\sigma}$ is a density operator, the reconstruction formula \eqref{reconstruction_state_D} becomes a particular instance of (\ref{eq:reconstruction_operator}):
\begin{equation}\label{rec_states_schur}
\boldsymbol{\sigma} = \frac{1}{|G|}\sum_{g\in G}\chi_\sigma (g)\, D(g) =  \frac{n}{|G|}\sum_{g\in G}\chi_\sigma (g)U(g)^\dagger.
\end{equation}

In the case of finite groups, the tomogram of the state $\rho$ can be obtained by using the discrete version of formula \eqref{Fourier_Transform_Tom_chap_2} (see also \cite{Ib11}). Let us transform $U(g)$ into a diagonal matrix $d_g$ by means of the unitary matrix $V_g$:
\begin{equation}
U(g)=V_gd_gV_g^\dagger,\qquad d_g=\textrm{diag}\left[\e^{i\theta_1(g)},\ldots,\e^{i\theta_n(g)}\right],
\end{equation}
then, compute the smeared character of $U(g)$:
\begin{equation}
\chi_\sigma (g)=\Tr\!\left(\boldsymbol{\sigma} V_gd_gV_g^\dagger\right)=\Tr\!\left(d_gV_g^\dagger\boldsymbol{\sigma} V_g\right)=\sum_{m=1}^n\e^{i\theta_m(g)}\left(V_g^\dagger\boldsymbol{\sigma} V_g\right)_{mm}.
\end{equation}
Therefore, the tomograms of the state $\boldsymbol{\sigma}$ are given by:
\begin{equation}\label{disc_tomogram}
\mathcal{W}_\sigma (m;g) = \left(V_g^\dagger\boldsymbol{\sigma} V_g\right)_{mm}.
\end{equation}
These tomograms define a stochastic vector, i.e.,
\begin{equation*}
\sum_{m=1}^n\mathcal{W}_\sigma (m;g)=\hspace{-0.1cm}\sum_{m=1}^n\left(V_g^\dagger\boldsymbol{\sigma} V_g\right)_{mm}\hspace{-0.25cm}=\hspace{0.1cm}\Tr(\boldsymbol{\sigma}) = 1\, ,
\end{equation*}
and they are non-negative:
\begin{equation*}
0\leq\mathcal{W}_\rho(m;g)\leq 1,\qquad m=1,\ldots,n,\qquad\forall g\in G,
\end{equation*}
since the density matrix $\boldsymbol{\sigma}$ is   positive semi-definite.
Thus, we have shown that the smeared characters can be obtained as a Discrete Fourier Transform of the tomograms \eqref{disc_tomogram}:
\begin{equation}
\chi_\sigma (g) = \sum_{m=1}^n\e^{i\theta_m(g)}\mathcal{W}_\sigma (m;g) \, ,
\end{equation}
and we obtain the following inverse Radon transform for quantum tomography on finite groups:
$$
\boldsymbol{\sigma}  =  \frac{n}{|G|}\sum_{g\in G}  \sum_{m=1}^n\e^{i\theta_m(g)}\mathcal{W}_\sigma (m;g) U(g)^\dagger \, .
$$
Let us consider now a subgroup $H\subset G$ of a finite or compact Lie group $G$. The restriction of the representation $U$ to the subgroup $H$, sometimes denoted by $U\!\downarrow\! H$ and called the subduced representation of $U$ to $H$, will be, in general, reducible even if $U$ is irreducible.

Let us suppose that the state $\rho$ satisfies the following orthogonality relations:
\begin{equation}\label{cond_adapted_states}
\Tr\!\big(\boldsymbol{\sigma} U(g)\big)=0,\qquad g\in G\setminus H,
\end{equation}
that is, the inner products with the unitary operators corresponding to the elements of $G$ not in the subgroup $H$ vanish. Therefore, in this case, we have similar formulas to \eqref{rec_states_schur} and \eqref{rec_states_compact_schur} even if the representation $U\!\downarrow\! H$ is reducible:
\begin{equation}\label{adapted_state_a}
\boldsymbol{\sigma} = \frac{n}{|G|}\sum_{g\in H}\chi_\sigma (g)U(g)^\dagger \, ,
\end{equation}
in the finite case, and
\begin{equation}\label{adapted_state_b}
\boldsymbol{\sigma} = n\hspace{-0.15cm}\int\limits_{\hspace{.3cm}H}\hspace{-0.1cm}\chi_\sigma (g)U(g)^\dagger\diff\mu(g) \, ,
\end{equation}
in the compact situation.

Such states will be said to be \textit{adapted states} to the subgroup $H$. These states have interesting properties because they share the same symmetry than the subgroup $H$. For example, they can be used to get the Clebsch--Gordan decomposition of a unitary reducible representation of a finite or compact Lie group (see the numerical algorithm presented in \cite{Ib17}).


\subsubsection{The regular representation}

Another instance that can be treated similarly is when we consider the regular representation of a finite group. The regular representation of the locally compact group $G$ is the unitary representation obtained from the action of the group $G$ on itself, in the Hilbert space of square integrable functions on the group, $\mathcal{H}=L^{2}(G,\mu_G)$, where $\mu_G$ denotes the left(right)-invariant Haar measure by left(right) translations.

The left regular representation $U^{reg}_L(h)$ is defined as follows:
\begin{equation}
\big(U^{reg}_L(h)\psi\big)(g)=\psi(h^{-1}g),\qquad \psi\in L^{2}(G,\mu_G),
\end{equation}
and the right regular representation is defined analogously.

If $G$ is finite, it is clear that $L^{2}(G)$ is isometrically isomorphic with the group algebra $\mathbb{C}[G]$:
\begin{equation}
\mathcal{H}\cong\mathbb{C}[G]=\Big\{|\alpha\rangle=\hspace{-0.05cm}\sum_{g\in G}\alpha_g|g\rangle\,\big|\,\alpha_g\in\mathbb{C}\Big\},
\end{equation}
with inner product $\displaystyle{\langle\alpha,\beta\rangle=\sum_{g\in G}\overline{\alpha}_g\beta_g}$. The action of the group is given by:
\begin{equation}
U_L^{reg}(h)|\alpha\rangle=\hspace{-0.05cm}\sum_{g\in G}\alpha_{h^{-1}g}|g\rangle=\hspace{-0.05cm}\sum_{g'\in G}\alpha_{g'}|hg'\rangle,
\end{equation}
then, we can interpret the left regular representation $U_L^{reg}$ as
\phantomsection\label{left_regular_rep_ps}\begin{equation}\label{left_regular_rep}
U_L^{reg}(h)|g\rangle=|hg\rangle,\qquad\forall g,h\in G.
\end{equation}

From the orthogonality relation satisfied by the regular representation:
\begin{equation}
\Tr\!\left(U_L^{reg}(g)^\dagger U_L^{reg}(g')\right)=n\delta_{g^{-1}g'},
\end{equation}
the character of the representation is easily computed:
\begin{equation}\label{regular_character}
\chi^{reg}(g)=n\delta_g=\left\{\begin{matrix}
n& g=e,\\
0&\mbox{otherwise\,,}
\end{matrix}\right.
\end{equation}
with $n=\dim\mathcal{H} = |G|$. In that case, the reconstruction formula is given by eq.\,\eqref{adapted_state_a}.

For compact groups, we have similar results, however the character $\chi^{reg}$ is now a Dirac delta distribution:
\begin{equation}
\chi^{reg}(g)=\delta(g),\qquad g\in G,
\end{equation}
and the theorem of Harish--Chandra (see \cite{Ar88}) allows to extend the result in eq.\,\eqref{regular_character} to semisimple Lie groups, in which case, the reconstruction formula is again Eq. \eqref{adapted_state_b} with $n=1$.


\subsection{Tomography with the Weyl-Heisenberg group}\label{resume_heis}

Another example, often discussed in applications because of its obvious experimental and historial origins, iss the Weyl-Heisenberg group which, contrary to the examples discussed in the previous section, is a nilpotent non-compact group.

Quantum state tomography, as discussed in the introduction, recall Eq. (\ref{eq:quantum_radon_heisenberg}), relies on the fact that
the position and momentum operators $\textbf{Q},\textbf{P}$, which satisfy the canonical commutation relation\footnote{In physical applications there is a dimensional constant $\hbar$ on the formula, that will be taken to be 1 in what follows.} $[\textbf{Q},\textbf{P}]=i\mathbbm{1}$, determine a realization of the Lie algebra of the Weyl-Heisenberg group (see for instance \cite{Ib09b}  and references therein).

Let $(V,\omega )$ be a symplectic, $2n$-dimensional real vector space, i.e., $\omega$ is a non-degenerate skew-symmetric form on $V$,  and
consider the $(2n+1)$-dimensional Weyl-Heisenberg group $W_n$
which is the central extension by the group $U(1)$  of the Abelian group $V$ with respect to the
2--cocycle defined by $\omega$, that is, the Weyl-Heisenberg group $W_n$
is the set of pairs $g = (\mathbf{v},u) \in V\times U(1)$, $\mathbf{v} \in V$, $u = e^{is} \in \mathbb{C}$,  with the following composition rule:
\begin{equation}
g\circ g^{\prime }= (\mathbf{v} ,u)\circ (\mathbf{v}^{\prime }, u^{\prime })=(\mathbf{v}+\mathbf{v}^{\prime
}, uu^{\prime }\e^{\frac{i}{2}\omega (\mathbf{v},\mathbf{v}^{\prime})}).
\end{equation}
The Haar measure in $W_n$ is given by the product measure of the standard Haar measure on $V$, i.e., the Lebesgue measure associated to the volume form $\omega^{2n}$,  and the biinvariant Haar measure on the group $U(1)$, that is, $ds/2\pi$.

The irreducible unitary representations of the Weyl-Heisenberg group $W_n$
may be constructed from a Weyl system on the symplectic space $(V,\omega )$,
that is, consider a strongly continuous map $W$ that associates to any vector $%
\mathbf{v}\in V$ a unitary operator $W(\mathbf{v})$ acting on a Hilbert space $%
\mathcal{H}$ and satisfying:
\begin{equation}
W(\mathbf{v})W(\mathbf{v}^{\prime })=W(\mathbf{v}+\mathbf{v}^{\prime })\exp
\left( \frac{i}{2}\omega (\mathbf{v},\mathbf{v}^{\prime })\right) ,
\label{weylComp}
\end{equation}
from which we get the Weyl exponentiated form of the canonical commutation relations:
\begin{equation}
W(\mathbf{v})W(\mathbf{v}^{\prime })=W(\mathbf{v}^{\prime })W(\mathbf{v}%
)\exp \left( i\omega (\mathbf{v},\mathbf{v}^{\prime })\right) .
\label{weylComm}
\end{equation}
Von Neumann theorem \cite{Es04} shows that it is always possible to
realize the Hilbert space $\mathcal{H}$ as the space of square integrable functions with support in $\mathcal{L}$ where $\mathcal{L}$ is any
Lagrangian subspace of $V$, i.e., $\mathcal{L}$ is a maximal isotropic subspace of $V$. The choice of a Lagrangian subspace induces the corresponding polarization, $V = \mathcal{L}
\oplus \mathcal{L}^*$, thus, any vector $\mathbf{v} \in V$, has the form $ \mathbf{v} = (\boldsymbol{\mu}, \boldsymbol{\nu}) \in \mathcal{L}\oplus \mathcal{L}^*$. 

The unitary operators $W(\mathbf{v})$ realizing the elements of the group are the usual displacement
operators: $(D(\boldsymbol{\mu},\boldsymbol{\nu})\psi) (\mathbf{x}) = \e^{i \boldsymbol{\mu}\cdot \mathbf{x}}\psi (\mathbf{x} + \boldsymbol{\nu})$, and a unitary
irreducible representation is provided by the expression:
\begin{equation}
U(g)=U(\boldsymbol{\mu},\boldsymbol{\nu},e^{is})=D(\boldsymbol{\mu},\boldsymbol{\nu}) e^{i sI} \, .
\label{reppol}
\end{equation}
or, even more explicitly, 
$$
(U(\boldsymbol{\mu},\boldsymbol{\nu},e^{is})\psi)(\mathbf{x}) = e^{is} \e^{i \boldsymbol{\mu}\cdot \mathbf{x}}\psi (\mathbf{x} + \boldsymbol{\nu}) \, ,
$$
for any function $\psi \in L^2(\mathcal{L}, \diff^n\nu)$.  Such representation is square integrable. In fact, for any $\psi \in L^2(\mathcal{L}, \diff^n \nu)$ the map $(\mathbf{x}, \boldsymbol{\nu})\mapsto \overline{\psi (\mathbf{x}+ \boldsymbol{\nu})}\psi(\mathbf{x})$ is measurable and, as it can be checked easily:
$$
\int_{\mathcal{L}}  \diff^n x \int_{\mathcal{L}} \diff^n \nu \,  |\overline{\psi (\mathbf{x}+ \boldsymbol{\nu})}\psi(\mathbf{x}) |^2 = || \psi ||^4 \, .
$$
Then, the map $(\boldsymbol{\nu}, \boldsymbol{\mu}, e^{is} ) \mapsto \langle U(\boldsymbol{\nu}, \boldsymbol{\mu}, e^{is} )\psi, \psi \rangle$, satisfies:
\begin{eqnarray*}
&& \frac{1}{2\pi}\int_{W_n} \diff^n \nu \, \diff^n \mu \, \diff s\, |\langle U(\boldsymbol{\nu}, \boldsymbol{\mu}, e^{is}) \psi, \psi \rangle|^2 = \\
&& =  \frac{1}{2\pi} \int_0^{2\pi} \diff s \, \int_{\mathcal{L}}\diff^n \nu   \int_{\mathcal{L}^*} \diff^n \mu\, \left| \int_{\mathcal{L}} \diff^n x \e^{-is} \e^{-i \boldsymbol{\mu}\cdot \mathbf{x}} \overline{\psi (\mathbf{x} + \boldsymbol{\nu})} \psi(\mathbf{x})\right|^2  = \\
&& =  \int_{\mathcal{L}}\diff^n \nu   \int_{\mathcal{L}^*} \diff^n \mu\, \left| \int_{\mathcal{L}} \diff^n x \e^{-i \boldsymbol{\mu}\cdot \mathbf{x}} \overline{\psi (\mathbf{x} + \boldsymbol{\nu})} \psi(\mathbf{x})\right|^2  = \\
&& =  \int_{\mathcal{L}}\diff^n \nu  \int_{\mathcal{L}} \diff^n x\, \left| \overline{\psi (\mathbf{x} + \boldsymbol{\nu})} \psi(\mathbf{x})) \right|^2  = || \psi||^4 \, ,
\end{eqnarray*}
and the representation is square integrable with virtual degree $d = 1$.

The self-adjoint operators $\mathbf{Q}$, $\mathbf{P}$ associated to the previous irreducible representation that provide the representation of the Lie algebra $\mathfrak{w}_n$ of the Weyl-Heisenberg group $W_n$ described in Eq. (\ref{selfadjoint_exp}) are given by the so called standard canonical quantization: 
$$
(\mathbf{Q} \psi) (\mathbf{x}) = \mathbf{x} \psi (\mathbf{x}) \, , \qquad (\mathbf{P} \psi) (\mathbf{x}) = -i \frac{\partial \psi}{\partial \mathbf{x} } \, ,
$$
and 
\begin{equation}\label{eq:QP}
U(\exp (t \xi)) = U(t\boldsymbol{\mu},t\boldsymbol{\nu}, e^{ist}) = \e^{ist I} \e^{it(\mu\textbf{Q}+\nu\textbf{P})} \, ,
\end{equation}
where $\xi\in \mathfrak{w}_n$ denotes a generic element on the Lie algebra of the Weyl-Heisenberg group, that is, $\xi = (\boldsymbol{\mu}, \boldsymbol{\nu}, s)$. The Lie algebra of the Abelian group $V$ is identified with $V$ itself (and the exponential map is the identity) and the Lie algebra of $U(1)$ is identified with $\mathbf{R}$ with exponential map $\exp (s) = e^{is}$.  In (\ref{eq:QP}), with a slight abuse of notation, we have indicated by $\mu$, $\nu$ the coordinates $\mu_k$, $\nu_k$, of the elements of the Lie algebra (i.e., vectors on the linear space $V$) $\boldsymbol{\mu}$, $\boldsymbol{\nu}$ respectively, so that $\mu \mathbf{Q} + \nu \mathbf{P} = \sum_{k =1}^n  \mu_k \mathbf{Q}_k + \nu_k \mathbf{P}_k$.

According with Thm. \ref{tom_cart},  Eq. (\ref{eq:TrsigmaX}), we obtain the tomogram for the state $\boldsymbol{\sigma}$:
$$
\mathcal{W}_{\sigma}(X;\xi)=\Tr\!\big(\boldsymbol{\sigma}\,\delta(X\mathbbm{1}-\langle\boldsymbol{\Theta},\xi\rangle)\big)
 \, ,
$$
for all $\xi = (\boldsymbol{\mu},\boldsymbol{\nu}, s) \in\mathfrak{w}_n$ and $X\in\mathbb{R}$, that again, using the more convenient notation $\mu$, $\nu$ for the coordinates of the elements on the Lie algebra and relabelling the variable $X-s$ as $X$, we get the well-known expression:
$$
\mathcal{W}_{\sigma}(X;\mu, \nu )= \Tr \! \big(\boldsymbol{\sigma}\,\delta( X\mathbbm{1}- \mu \mathbf{Q} -  \nu \mathbf{P})\big) \, ,
$$
that was discussed in the Introduction, Eq. (\ref{eq:quantum_radon_heisenberg}).

Finally, the reconstruction formula, Eq. (\ref{eq:reconstruction_operator}), for the state $\boldsymbol{\sigma}$ reads as
\begin{equation}\label{reconstruction_rho_HW}
\boldsymbol{\sigma} = \frac{1}{2\pi}\hspace{-0cm}\int_{W_n} \chi_\sigma (\mu,\nu, s) \, U(\mu,\nu, e^{is})^\dagger\diff^n\mu\, \diff^n \nu \, \diff s \, ,
\end{equation}
but then, because of Eq. (\ref{smeared_character_Fourier_new}), we get:
\begin{eqnarray}
\boldsymbol{\sigma} & = & \frac{1}{2\pi}\hspace{-0cm}\int_{W_n} \diff^n\mu\, \diff^n \nu \, \diff s \int\limits_{\hspace{-.04cm}-\infty}^{\hspace{.45cm}\infty}\hspace{-0.15cm} \diff X\,  \mathcal{W}_\sigma (X-s, \mu,\nu ) \, U(\mu,\nu, e^{is})^\dagger
\e^{itX} \, ,
\\ & = &\int\limits_{\hspace{.4cm}\mathbb{R}^{2n +1}}\hspace{-0.17cm}\mathcal{W}_\sigma (X,\mu,\nu)\e^{i(X \mathbbm{1} - \mu\textbf{Q}-\nu\textbf{P})}\, \diff X\, \diff^n\mu\, \diff^n\nu\, . \nonumber 
\end{eqnarray}

If we compare this equation with \eqref{Inverse_Radon_Transform_2} we see that are similar only except that the role of the probability density $f(q,p)$ is played now by the density operator $\boldsymbol{\sigma}$ and the role of classical position and momentum $q$ and $p$ is played by the operators $\textbf{Q}$ and $\textbf{P}$.


\subsection{Tomograms of an ensemble of quantum harmonic oscillators}\label{section_Toms_q_h_o}

As an application of the quantum Radon transform we will exhibit the tomograms of pure states of an ensemble of quantum harmonic oscillators described by the Hamiltonian operator:
\begin{equation}
\textbf{H}=\sum_{k=1}^n\omega_ka^\dagger_ka_k+\frac{1}{2}\sum_{k=1}^n\omega_k,
\end{equation}
with $a_k$, $a_k^\dagger$ standard creation and annihilation operators satisfying the canonical commutation relations:
\begin{equation}\label{Commutator_creation_annihilation_f}
\big[a_k,a_{k'}^\dagger\hspace{-0.cm}\big]=\delta_{kk'},\quad\big[a_k,a_{k'}\big]=\big[a_k^\dagger,a_{k'}^\dagger\big] = 0\, .
\end{equation}
Let $\boldsymbol{\rho}$ be the pure state corresponding to the system in which each particle has momentum $k_i$, $i=1,\ldots,n$:
\begin{equation}
\boldsymbol{\rho}=|1_{k_1},\ldots,1_{k_n}\rangle\langle1_{ k_1},\ldots,1_{k_n}|.
\end{equation}
Recall that the annihilation and creation operators act on the ground state $|0,\ldots,0\rangle$ as:
\begin{equation}
a_j^\dagger|0,\ldots,0\rangle=|0,\ldots,\overset{j}{1},\ldots,0\rangle,\qquad a_j|0,\ldots,0\rangle=0.
\end{equation}
The tomogram of the state $\boldsymbol{\rho}$ associated to the representation of the Weyl-Heisenberg group $W_n$ discussed in the previous section, is given by:
\begin{equation}
\mathcal{W}_{\boldsymbol{\rho}}(X,\mu,\nu) = \Tr\!\big(\boldsymbol{\rho}\, \delta(X\mathbbm{1}- \mu\textbf{Q}-\nu\textbf{P})\big) \, .
\end{equation}
Introducing holomorphic variables $w_j=\frac{\mu_j+i\nu_j}{\sqrt{2}}$, we get:
\begin{eqnarray}
\mathcal{W}_{\boldsymbol{\rho}}(X,w,\overline{w}) &=& \Tr\!\big(\boldsymbol{\rho}\, \delta(X-\overline{\boldsymbol{w}}\cdot\boldsymbol{a}-\boldsymbol{w}\cdot\boldsymbol{a}^\dagger)\big) \nonumber \\
&=& \vphantom{\sum^n}\langle 0|a_1\,\cdots\, a_n \, \delta(X-\overline{\boldsymbol{w}}\cdot\boldsymbol{a}-\boldsymbol{w}\cdot\boldsymbol{a}^\dagger) \, a_n^\dagger\,\cdots\, a_1^\dagger|0\rangle \\
&=& \frac{1}{2\pi} \int\limits_{\hspace{-.04cm}-\infty}^{\hspace{.45cm}\infty}\hspace{-0.15cm}\e^{ikX}\e^{-k^2|\boldsymbol{w}|^2/2}                                                                                                                                                                                                                                                                                     \langle 0|a_1\,\cdots\, a_n \e^{-ik(\overline{\boldsymbol{w}}\cdot\boldsymbol{a}+\boldsymbol{w}\cdot\boldsymbol{a}^\dagger)}a_n^\dagger\,\cdots\, a_1^\dagger|0\rangle\diff k\, . \nonumber
\end{eqnarray}

From the canonical commutation relations \eqref{Commutator_creation_annihilation_f} we get $\left[\vphantom{a^\dagger}\right.\hspace{-0.1cm}a^n_k,a^\dagger_j\hspace{-0.1cm}\left.\vphantom{a^\dagger}\right]=na^{n-1}_k\delta_{kj}$, therefore:
\begin{equation}
\langle 0|a_k\e^{-ikw_ia^\dagger_j}=\langle 0|(a_k-ikw_i\delta_{kj}),\qquad
\e^{-ik\overline{w}_ia_j}a^\dagger_k|0\rangle=(a^\dagger_k-ik\overline{w}_i\delta_{kj})|0\rangle.
\end{equation}
Hence, using this result and the Baker-Campbell-Hausdorff formula repeatedly, we get:
\begin{multline}\label{Tomogram_q_h_oscillator}
\mathcal{W}_{\boldsymbol{\rho}}(X,\boldsymbol{w},\overline{\boldsymbol{w}})=\frac{1}{2\pi}\hspace{-.15cm}\int\limits_{\hspace{-.04cm}-\infty}^{\hspace{.45cm}\infty}\hspace{-0.15cm}\e^{ikX}\e\begin{matrix}
                                                                                                                                                                                                                                                                                                    \hspace{-.1cm}\vspace{-.1cm}\scriptstyle{-k^2|\boldsymbol{w}|^2} \\
                                                                                                                                                                                                                                                                                                    \vspace{0cm} \\
                                                                                                                                                                                                                                                                                                  \end{matrix}\begin{matrix}
                                                                                                                                                                                                                                                                                                    \hspace{-0.03cm}\vspace{-.2cm}\scriptstyle{/} \\
                                                                                                                                                                                                                                                                                                    \vspace{0cm} \\
                                                                                                                                                                                                                                                                                                  \end{matrix}\begin{matrix}
                                                                                                                                                                                                                                                                                                    \hspace{0cm}\vspace{-.3cm}\scriptstyle{2} \\
                                                                                                                                                                                                                                                                                                    \vspace{0cm} \\
                                                                                                                                                                                                                                                                                                  \end{matrix}(1-k^2|w_1|^2)\cdots(1-k^2|w_n|^2)\diff k\\
                                                                                                                                                                                                                                                                                                  =\frac{1}{\sqrt{\pi(\boldsymbol{\mu}^2+\boldsymbol{\nu}^2)}}\left(1+\alpha_1\frac{\diff^2}{\diff X^2}
                                                                                                                                                                                                                                                                                               +\alpha_{n}\frac{\diff^{2n}}{\diff X^{2n}}\right)\e^{-\frac{X^2}{\boldsymbol{\mu}^2+\boldsymbol{\nu}^2}},                                                                                                                                                                                                                                                                                                                                                                                                                                                                                                                                                                       
\end{multline}
where
\begin{align}
&\alpha_1=\sum_{i_1=1}^n|w_{i_1}|^2=|\boldsymbol{w}|^2=2^{-1}(\boldsymbol{\mu}^2+\boldsymbol{\nu}^2),\nonumber\\
&\hspace{-0.1cm}\alpha_2\hspace{0.1cm}=\hspace{-0.3cm}\sum_{\hspace{0.2cm}i_1,\,i_2>i_1}^n\hspace{-0.33cm}|w_{i_1}|^2|w_{i_2}|^2=2^{-2}\hspace{-0.4cm}\sum_{\hspace{0.2cm}i_1,\,i_2>i_1}^n\hspace{-0.3cm}(\mu_{i_1}^2+\nu_{i_1}^2)(\mu_{i_2}^2+\nu_{i_2}^2)\,  , \cdots\nonumber\\
&\hspace{-0.5cm}\alpha_{n-1}\hspace{0.1cm}=\hspace{-0.4cm}\sum_{\substack{i_1,\,i_2>i_1,\ldots,\\i_{n-1}>\cdots>i_1}}^n\hspace{-0.5cm}|w_{i_1}|^2\cdots|w_{i_n}|^2=2^{-(n-1)}\hspace{-0.6cm}\sum_{\substack{i_1,\,i_2>i_1,\ldots,\\i_{n-1}>\cdots>i_1}}\hspace{-0.25cm}(\mu_{i_1}^2+\nu_{i_1}^2)\cdots(\mu_{i_n}^2+\nu_{i_n}^2),\nonumber\\
&\alpha_n=|w_1|^2\cdots\,|w_n|^2=2^{-n}(\mu_1^2+\nu_1^2)\,\cdots\,(\mu_n^2+\nu_n^2).
\end{align}
Thus, using the Hermite polynomials:
\begin{equation}\label{Hermite_pols}
H_n(x)=(-1)^n\e^{x^2}\frac{\diff^n}{\diff x^n}\e^{-x^2},
\end{equation}
we finally obtain:
\begin{multline}\label{NDimqharm}
\mathcal{W}_{\boldsymbol{\rho}}(X,\boldsymbol{\mu},\boldsymbol{\nu})=\frac{1}{\sqrt{\pi(\boldsymbol{\mu}^2+\boldsymbol{\nu}^2)}}\left[1+\frac{\alpha_1}{\boldsymbol{\mu}^2+\boldsymbol{\nu}^2}H_2\left(\frac{X}{\sqrt{\boldsymbol{\mu}^2+\boldsymbol{\nu}^2}}\right)+\cdots \right.\\
\cdots+\left.\frac{\alpha_{n}}{(\boldsymbol{\mu}^2+\boldsymbol{\nu}^2)^n}H_{2n}\left(\frac{X}{\sqrt{\boldsymbol{\mu}^2+\boldsymbol{\nu}^2}}\right)\right]\e^{-\frac{X^2}{\boldsymbol{\mu}^2+\boldsymbol{\nu}^2}}.
\end{multline}

\section*{Acknowledgments} A.I. acknowledges financial support from the Spanish Ministry of Economy and Competitiveness, through the Severo Ochoa Programme for Centres of Excellence in RD (SEV-2015/0554).
A.I. would like to thank partial support provided by the MINECO grant MTM2017-84098-P, and QUITEMAD++, S2018/TCS-A4342.

\medskip
Received xxxx 20xx; revised xxxx 20xx.
\medskip

\end{document}